\documentclass[journal, twoside, web]{ieeecolor}
\usepackage{generic}
\usepackage{cite}
\usepackage{amsmath,amssymb,amsfonts}
\usepackage{algorithmic}
\usepackage{graphicx}
\usepackage{algorithm,algorithmic}
\usepackage{hyperref}
\usepackage{textcomp}

\def\BibTeX{{\rm B\kern-.05em{\sc i\kern-.025em b}\kern-.08em
    T\kern-.1667em\lower.7ex\hbox{E}\kern-.125emX}}

\newcommand{\Aut}{\text{Aut}}
\newcommand{\aut}{\mathfrak{aut}}
\newcommand{\Inn}{\text{Inn}}
\newcommand{\inn}{\mathfrak{inn}}
\newcommand{\Out}{\text{Out}}

\newcommand{\GL}{\text{GL}}
\newcommand{\SO}{\text{SO}}
\newcommand{\SIM}{\text{SIM}}
\newcommand{\TFG}{\text{TFG}}
\newcommand{\tfg}{\mathfrak{tfg}}
\newcommand{\tr}{\text{tr}}
\newcommand{\hatbar}[1]{\hat{\bar{#1}}}

\newtheorem{assumption}{Assumption}
\newtheorem{remark}{Remark}

\newtheorem{proposition}{Proposition}
\newtheorem{lemma}{Lemma}
\newtheorem{theorem}{Theorem}

\def\ps@titlepagestyle{
  \def\@oddfoot{}\def\@evenfoot{}

  \if@confmode
    \def\@oddhead{}\def\@evenhead{}
  \else 
    \def\@oddhead{
      \vbox{
        \hbox to \textwidth{\scriptsize\hfil\leftmark\quad\thepage}
        \vskip 2pt
        \hrule height 0.4pt depth 0pt width \textwidth
      }
    }
    
    \def\@evenhead{
      \vbox{
        \hbox to \textwidth{\scriptsize\thepage\quad\rightmark\hfil}
        \vskip 2pt
        \hrule height 0.4pt depth 0pt width \textwidth
      }
    }
  \fi
  
  \if@draftclsmode
     \def\@oddhead{
       \vbox{
         \hbox to \textwidth{\scriptsize\leftmark \hfil \thepage}
         \vskip 2pt
         \hrule height 0.4pt depth 0pt width \textwidth
       }
     }
     \def\@evenhead{
       \vbox{
         \hbox to \textwidth{\scriptsize\thepage \hfil \leftmark}
         \vskip 2pt
         \hrule height 0.4pt depth 0pt width \textwidth
       }
     }
     \def\@oddfoot{\hfill\makebox[0pt][l]{\scriptsize DRAFT \today}\hfill}
     \def\@evenfoot{\hfill\makebox[0pt][r]{\scriptsize DRAFT \today}\hfill}
  \fi
}

\makeatother

\begin{document}

\title{Global Observer Design for a Class of Linear Observed Systems on Groups}
\author{Changwu Liu and Yuan Shen, \IEEEmembership{Senior Member, IEEE}
\thanks{The authors are with Department of Electronic Engineering, Tsinghua University, 100084 Beijing, China (e-mail: liucw\_ee@tsinghua.edu.cn; shenyuan\_ee@tsinghua.edu.cn).}
}

\maketitle

\begin{abstract}
    Linear observed systems on groups encode the geometry of a variety of practical state estimation problems. In this paper, we propose an observer framework for a class of linear observed systems by restricting a bi-invariant system on a Lie group to its normal subgroup. This structural property enables a system embedding of the original system into a linear time-varying system. An observer is constructed by first designing a Kalman-like observer for the embedded system and then reconstructing the group-valued state via optimization. Under an extrinsic observability rank condition, global exponential stability (GES) is achieved provided that one global optimum of the reconstruction optimization is found, reflecting the topological difficulties inherent to the non-Euclidean state space. Semi-global stability is guaranteed when input biases are jointly estimated. The theory is applied to the GES observer design for two-frame systems, capable of modeling a family of navigation problems. Simulations are provided to illustrate the implementation details. 
\end{abstract}


\section{Introduction}
\label{sec::introduction}

\IEEEPARstart{L}{inear} observed systems on groups are systems whose flows are compatible with group automorphisms \cite{LoS, PreIntMath, LoS_Mfd}. State estimation for those systems is of particular interest to the robotics and automation community, as it captures the geometry of many practical problems \cite{TFG, InEKF, AnnuRevInEKF, filtering_lie_groups}. Prototypical examples include the navigation of rigid bodies using multi-sensor information \cite{TFG,rotating_earth, EqVIO, MSCEqF, InGVIO, ExploitSymm, InEKF_wheel_gnss_arm, IEKF-wind}, i.e., estimating the attitudes, positions and linear velocities of some moving objects. Those variables are naturally combined together as transformations between coordinate frames. The set of transformations is a closed subgroup of the general linear group and hence is endowed with a natural Lie group structure. The success of these applications is attributed to respecting geometry, i.e., utilizing the group formulation of the state space as well as system properties mainly linked to linear observed structures.

An observer, or state estimator, is a dynamical system driven by known inputs and measurements \cite{obs_design_nls}. Asymptotic stability is the central objective in observer design. The main difficulties hindering the guarantee of stability for observers for general linear observed systems on groups are two-fold: the nonlinearities of the system equations and the non-Euclidean state space topology. The first category of observers for linear observed systems is based on linearization that respects the geometry. The invariant filter (InEKF) linearizes the system in exponential coordinates to obtain state-independent error dynamics, leveraging the linear observed structure. With gains computed from a Riccati equation akin to that in the Kalman filter using error dynamics, the InEKF achieves local stability under conditions similar to those of Kalman filters \cite{InEKF, AnnuRevInEKF}. Equivariant filter (EqF) considers a wider range of systems on homogeneous spaces, encompassing linear observed systems on groups. The EqF lifts the system to its symmetry group acting on the base manifold and designs an observer utilizing the invariant group error \cite{EqF, Annurev_EqF}. Again, its gain is obtained from a Riccati equation whose coefficients follow from linearization in the coordinates of the base manifold. Although InEKF and EqF are designed for general systems, their stability domains are only local due to linearization. The second category of observers is characterized by constructive approaches that achieve almost-global stability \cite{synchronous_observer_design, AutoErrorGroupAffine}. This is the best one could achieve when designing the correction using a continuous vector field, originating from topological obstructions \cite{topological_obstruction}, when the state space contains a non-contractible component. To the best of our knowledge, there is no out-of-the-box, easy-to-implement observer that achieves global exponential stability (GES) for general linear observed systems on groups.

There are numerous case-by-case studies for global observer design in rigid-body navigation. In addition to classical constructive almost-globally stable observers for attitude estimation \cite{cf_att,improved_cf_att}, recently proposed hybrid observers for IMU-based navigation with landmark or vision-type measurements are discussed in \cite{hybrid_observer_landmark, hybrid_observer_vision_aided, ins_observer_tutorial}. Such observers consist of continuous flows and discrete jumps and achieve GES, thereby overcoming topological obstructions. This construction relies heavily on the matrix group representation of IMU dynamics and corresponding innovations, both of which are closely related to the linear observed structure. In contrast to these constructive methods, switching to robocentric coordinates using body-referenced linear velocities and landmark positions simplifies the system model, yielding a linear time-varying (LTV) system \cite{att_scalar,body_frame_state_estimation,geometric_approach_imu_body, ins_ltv}. A Riccati observer \cite{riccati_observer_pnp} is designed for these auxiliary LTV systems, and the original states are reconstructed subsequently. Strong GES guarantees are obtained under persistent excitation. The above methods focus on specific examples, representing only the tip of the iceberg of the systems capable of being modeled by linear observed formulation. Moreover, the success of these key techniques implicitly relies on linear observed structures. The feasibility of these techniques for general linear observed systems on groups has not been studied.

In this paper, we delve into the linear observed structure and reveal a powerful connection between this structure and the possibility of an LTV embedding of linear observed systems on groups. A global observer framework follows naturally. Our contributions are summarized below.

\begin{itemize}
    \item Group-theoretic conditions allowing LTV embedding of linear observed systems on groups are established. The class of systems arising from the restriction of a bi-invariant system to any normal subgroup of the state space can be embedded into LTV systems.
    \item An out-of-the-box observer is proposed for embeddable linear observed systems on groups. If an extrinsic observability rank condition is satisfied, GES is achieved provided that one global optimum of the related optimization on the group can be found, reflecting the topological difficulties. The gap in observability assumptions required to achieve stability compared with the InEKF is studied. Joint estimation of input bias is also tackled.
    \item The observer framework is applied to two-frame systems \cite{TFG}. Implementations of GES observers for embeddable two-frame systems are provided with additional extension to bearing and range measurements.
    \item Simulations are provided to illustrate tuning and performance compared with the InEKF.
\end{itemize}

For the rest of the paper, Section~\ref{sec::preliminaries} reviews mathematical preliminaries. The properties allowing LTV embedding are discussed in Section~\ref{sec::properties_allowing_ltv_embedding} after illustrative introductory examples in Section~\ref{sec::intro_examples}. Our GES observer is detailed in Section~\ref{sec::global_observer_using_ltv_embedding}. A prototypical application of the theory to embeddable two-frame systems is discussed in Section~\ref{sec::two_frame_systems}. Simulations are provided in Section~\ref{sec::simulation_examples}.

\section{Preliminaries}
\label{sec::preliminaries}

\subsection{Notations}

$\mathbb{R}^n$ denotes an $n$-dimensional $\mathbb{R}$-vector space and $\mathbb{S}^{n}$ is the unit $n$-sphere. Lowercase and capital letters represent vectors and matrices respectively. We use $\Vert\cdot\Vert$ for the Euclidean norm of a vector or the Frobenius norm of a matrix. $\preceq,\ \succeq,\ \prec,\ \succ$ denote the partial order on symmetric matrices. $\mathbb{S}^n_+$ is the cone of symmetric positive definite $\mathbb{R}^{n\times n}$-matrices. $\hat{(\cdot)}$ and $(\cdot)$ denote the estimated and true states, respectively. $\text{diag}(\cdot)$ denotes the diagonal or block-diagonal matrix.

\subsection{Group Theory Basics}

We are interested in matrix Lie groups, namely the closed subgroups of $\GL(n,\mathbb{R})$, which are invertible $\mathbb{R}^{n\times n}$ matrices under multiplication. Let $\Aut(G),\Inn(G),\Out(G)$ be the automorphism, inner-automorphism and outer-automorphism groups, respectively. The rotation group in $\mathbb{R}^d$ is $\SO(d)=\left\{R\in\GL(d,\mathbb{R})\left|R^\top R=I_{d\times d},\det(R)=1\right.\right\}$. $\mathbb{R}^{n\times n}$ provides a global embedding for matrix groups as well as for their Lie algebras. Let $\mathcal{L}_{\mathfrak{g}}:\mathbb{R}^{\dim\mathfrak{g}}\rightarrow\mathfrak{g}\subset\mathbb{R}^{n\times n}$ be the Lie algebra isomorphism between a vector space and its matrix representation, e.g., the skew-symmetric matrix $\mathcal{L}_{\mathfrak{so}(3)}(x_1)x_2:=x_1\times x_2$ for all $x_1,x_2\in\mathbb{R}^3$ is defined by the $\mathbb{R}^3$ cross product. We also use the compact notation $(\cdot)^\times:=\mathcal{L}_{\mathfrak{so}(d)}(\cdot)$. Multiplication of a matrix group on a vector constitutes a linear action of the group under its natural representation. The reader is referred to \cite{GTM218, GTM222} for a thorough mathematical preparation.

\subsection{Error-State Extended Kalman Observer}\label{subsec::ekf}

Let $\dot{x}=f(x,u),\ y=h(x)$ be a system defined on $\mathbb{R}^n$. $f$ and $h$ are mappings with sufficient smoothness. A continuous-time error-state extended Kalman observer for this system is
\begin{equation}
    \dot{\hat{x}}=f(\hat{x},u)+K(y-h(\hat{x})),
\end{equation}
where $\hat{x}$ is the state estimate and $K$ is the gain. The Kalman methodology suggests $K=PH^\top\mathcal{R}^{-1}$, calculated from the Riccati ODE, where all matrices depend smoothly on time:
\begin{equation}
    \dot{P}=FP+PF^\top+\mathcal{G}Q\mathcal{G}^\top-PH^\top\mathcal{R}^{-1}HP.
\end{equation}
Although the observer herein is fully deterministic, a stochastic interpretation is helpful for its tuning. To be precise, let $n_1$ and $n_2$ be white Gaussian noises associated with the input $u$ and the output $y$, respectively. We wish to linearize the system around the current estimate $\hat{x}$. Define $x=\hat{x}+\delta x$, where $\delta{x}$ is considered an error state. We substitute $x$ into $\dot{x}=f(x,u+n_1),\ y=h(x)+n_2$, and obtain a first-order approximation
\begin{align}
    \delta\dot{x}&=\frac{\partial f}{\partial x}(\hat{x},u)\delta x+\frac{\partial f}{\partial u}(\hat{x},u)n_1\\
    y-h(\hat{x})&=\frac{\partial h}{\partial x}(\hat{x})\delta{x}+n_2.
\end{align}
The configuration for the observer is set to
\begin{equation}
    F:=\frac{\partial f}{\partial x}(\hat{x}, u),\ \mathcal{G}:=\frac{\partial f}{\partial u}(\hat{x}, u),\ H:=\frac{\partial h}{\partial x}(\hat{x}),
\end{equation}
and $Q:=\text{cov}(n_1)$ and $\mathcal{R}:=\text{cov}(n_2)$.

In later sections, we specify such an observer by providing formulas for the tuple $(F,\mathcal{G},H)$.

\subsection{Observability and Matrix Riccati/Lyapunov Equations}

Consider the LTV system $\dot{x}_{t}=A_tx_{t}+B_tu_{t},\ y_{t}=H_tx_{t}$ where the matrices and vectors are of compatible dimensions. Let $\Phi(t_2,t_1)$ be the state transition matrix satisfying $\frac{\partial\Phi(t_2,t_1)}{\partial t_2}=A_{t_2}\Phi(t_2,t_1)$ with initial condition $\Phi(t_1,t_1)=I$. The observability Gramian \cite{linear_systems_theory, riccati_observer_pnp} is $\mathcal{O}(t_2,t_1):=\int_{t_1}^{t_2}\Phi^\top(\tau,t_1)H_{\tau}^\top \mathcal{R}_\tau^{-1}H_\tau\Phi(\tau, t_1)d\tau$. The system is uniformly observable if there exist constants $\delta, \alpha>0$ such that $\mathcal{O}(t+\delta,t)\succeq\alpha I$ holds for every $t\in\mathbb{R}$ \cite{linear_systems_theory}. If the LTV system is uniformly controllable and observable, then the eigenvalues of $P$ are uniformly lower and upper bounded \cite{boundedness_riccati_ode}, i.e., there exist $p_m,p_M>0$ such that $p_mI\preceq P_t\preceq p_MI$ for all $t$. Uniform observability is related to the lower bound $p_m$ specifically.

\section{Introductory Examples}\label{sec::intro_examples}

Examples are provided to illustrate the motivation and pipeline of the proposed embedding-based observer. Two central problems are the LTV embeddability of our nonlinear system and the ability to reconstruct the original state from the estimate of the embedded linear system. The former is closely linked to a specific structure of the system, and the latter is connected to the observability of the system.

\subsection{Embedding-based Observer for LTI Systems}\label{subsec::embedding_lti}

Let $x_t\in\mathbb{R}^n,\ u_t\in\mathbb{R}^m,\ y_t\in\mathbb{R}^l$ be the state, the input and the output, respectively. Consider an LTI system in the form of $\dot{x}_t=Ax_t+Bu_t,\ y_t=Hx_t$, where the matrices are of compatible dimensions and the subscript $t$ indicates that the corresponding variable smoothly depends on time. Let $z_j=HA^jx\in\mathbb{R}^n$ for $j\in\mathbb{N}$. We have an inductive ODE as
\begin{equation}
    \label{eq::lti_zj1}
    \dot{z}_j=z_{j+1}+HA^jBu_t.
\end{equation}
This inductive definition of $z_j$ will terminate after finitely many steps by virtue of the Cayley-Hamilton theorem applied to the linear operator $A$, as
\begin{equation}
    z_{n}=HA^{n}x=H\left(\sum_{l=0}^{n-1}a_lA^l\right)x=\sum_{l=0}^{n-1}a_lz_l,
\end{equation}
where $a_l\in\mathbb{R}$ come from the coefficients of the characteristic polynomial of $A$. Hence, the construction terminates at
\begin{equation}
    \label{eq::lti_zj2}
    \dot{z}_{n-1}=\sum_{l=0}^{n-1}a_lz_l+HA^{n-1}Bu_t.
\end{equation}
This $n^2$-dimensional system \eqref{eq::lti_zj1} and \eqref{eq::lti_zj2} with output $y=z_0$ is linear and fully observable, because it's in a companion form. Moreover, this system has the identical input-output behavior as the original system regarding $x$. We call such a system an embedding of the original system. Unlike immersion, which is valid only in an open neighborhood, the embedding here is a global concept in the entire state space. Immersibility is traditionally related to an observability rank condition of the system \cite{nonlinear_immersion2}. In contrast, the embedding of an LTI system is a consequence of the linear structure and does not require any observability a priori.

It's natural to apply a linear observer to obtain an estimate $\hat{z}$ of the embedded state. We are left to reconstruct the original state from $\hat{z}$, but this is just solving a linear equation
\begin{equation}
    \begin{bmatrix}
        \hat{z}_0\\
        \hat{z}_1\\
        \vdots\\
        \hat{z}_{n-1}
    \end{bmatrix}=\begin{bmatrix}
        H\\
        HA\\
        \vdots\\
        HA^{n-1}
    \end{bmatrix}x,
\end{equation}
whose coefficient matrix is exactly the Kalman observability matrix. Unique reconstruction of $x$ is possible if and only if the original system is observable.

In summary, two key aspects of the embedding-based observer for LTI systems are
\begin{enumerate}
        \renewcommand{\labelenumi}{(\arabic{enumi})}
        \item the linear structure of the drift vector field $Ax$ leads to a global embedding into a fully observable linear system by applying the Cayley-Hamilton theorem on $A$;
        \item the observability is transferred into the ability to reconstruct the original state from the embedding.
\end{enumerate}

A linear observed system with state $\chi$ on a group $G$ 
\begin{alignat}{3}
    \label{eq::los_lti1}
    \dot{\chi}_t=&f(\chi_t)+\sum_{l=1}^n &&u_lY_l(\chi_t),\quad y_t=&&h(\chi_t),\\
    \notag
    &\Updownarrow &&\Updownarrow && \Updownarrow\\
    \label{eq::los_lti2}
    \dot{x}_t=&Ax_t\quad\ +&&Bu_t,\qquad\quad y_t=&&Hx_t,
\end{alignat}
generalizes linear systems on $\mathbb{R}^n$ in the following aspects. 
\begin{enumerate}
    \renewcommand{\labelenumi}{(\arabic{enumi})}
    \item The vector fields corresponding to the drift term constitute the Lie algebra of the automorphism of the state space. To be precise, $\Aut(\mathbb{R}^n)\cong\GL(n,\mathbb{R})$, and $A\in\mathbb{R}^{n\times n}\cong\mathfrak{gl}(n,\mathbb{R})$. Similarly, we have $f(\chi_t)\in\aut(G)$, which is the Lie algebra of $\Aut(G)$.
    \item The vector fields corresponding to the input are invariant vector fields on the state space. $Bu_t$ is constant on $\mathbb{R}^n$, and thus invariant. The $Y_j(\chi_t)$ are left- or right-invariant.
    \item The outputs are group-morphisms to a vector space. Specifically, $Hx$ is linear in $x$, and $h(\chi)$ is defined as a linear action on a known vector.
\end{enumerate}

Mimicking the linear mechanisms, we wish to
\begin{enumerate}
    \renewcommand{\labelenumi}{(\arabic{enumi})}
    \item exploit the linear observed structure and embed \eqref{eq::los_lti1} into a linear system using the Cayley-Hamilton theorem for some linear operator related to $\aut(G)$;
    \item establish an observability condition to reconstruct the group-valued state from the embedded state. 
\end{enumerate}

It is the main goal of the following sections to provide rigorous technical details for these analogies. Before delving into embeddability details, we demonstrate this pipeline using two classical examples of linear observed systems.

\subsection{Embedding-based Attitude Observer}
\label{subsec::att_observer}

Let $R_t\in\SO(3)$ be the state. Let $y_t^{(i)}\in\mathbb{R}^3$ be the $i$-th output and $d_i\in\mathbb{R}^3$ be some known vector. The system equations are
\begin{equation}
    \dot{R}_t=R_t\omega_t^\times,\ \ y_t^{(i)}=R_t^{-1}d_i,\ \ i=1,2,...,N,
\end{equation}
where $\omega_t\in\mathbb{R}^3$ is the gyroscope input. Defining $z_0^{(i)}=R_t^{-1}d_i$ for $1\le i\le N$, we obtain the embedded LTV system as
\begin{equation}
    \label{eq::emb_sys_att_obs}
    \left\{\begin{aligned}
        \dot{z}_0^{(i)}&=-(R_t^{-1}\dot{R}_t)(R_td_i)=-\omega_t^\times z_0^{(i)}\\
        y_t^{(i)}&=z_0^{(i)}
    \end{aligned}
    \right..
\end{equation}
It is clear that \eqref{eq::emb_sys_att_obs} is (Kalman) observable for any $N$. Suppose we apply a Kalman filter to \eqref{eq::emb_sys_att_obs} and obtain estimates of $z_0^{(i)}$, we then reconstruct an estimate of the attitude by
\begin{equation}
    \label{eq::opt_att}
    \hat{R}=\mathop{\arg\min}_{\chi\in\SO(3)}\Vert \hat{Z}-\chi^{-1} D\Vert^2=VSU^\top,
\end{equation}
where $\hat{Z}=[\hat{z}_0^{(1)},...,\hat{z}_0^{(N)}]\in\mathbb{R}^{3\times N}$ and $D=[d_1,...,d_N]\in\mathbb{R}^{3\times N}$. Using the Umeyama algorithm \cite{umeyama_algorithm}, \eqref{eq::opt_att} can be solved in closed form as $\hat{R}=VSU^\top$,
where $\hat{Z}D^\top=U\Lambda V^\top$ is a singular value decomposition. $\Lambda\in\mathbb{R}^{3\times 3}$ is diagonal, and $U,V\in\text{O}(3)$. $S=\text{diag}(1,1,\det(UV))$. In practice, the optimization cost may be weighted for scaling reasons.

To guarantee the uniqueness of the solution $\hat{R}=VSU^\top$ of \eqref{eq::opt_att}, we require $\text{rank}(D)=3$, which is the number of columns of $R$. This means we need at least three linearly independent vectors $d_i$. In general, the state reconstruction requirement is slightly stronger than the intrinsic observability assumption in invariant filtering, where we require only two linearly independent vectors $d_i$. Fortunately, a remedy is possible in our embedding-based observer if we only have two linear-independent vectors $d_1$ and $d_2$. We can generate a third output equation using the cross product, i.e., $y_1^\times y_2=R_t^{-1}d_1^\times d_2$, and hence we have $D=[d_1,d_2,d_1^\times d_2]$, whose column rank is full. In the next sections, we analyze the observability gap between our extrinsic embedding-based observer and the intrinsic implementation using invariant filters. This observability gap does not appear in every application, and we provide conditions for such a gap to vanish.

\subsection{Embedding-based IMU-Landmark Pose Observer}
\label{subsec::imu_lmk_observer}

Let $R_t\in\SO(3)$, $p\in\mathbb{R}^3$ and $v\in\mathbb{R}^3$ be the attitude, position and velocity, respectively. The IMU dynamics are
\begin{equation}
    \dot{R}_t=R_t\omega^\times,\ \dot{p}_t=v_t,\ \dot{v}_t=R_ta_t+g,
\end{equation}
where $\omega_t$ and $a_t$ are the gyroscope and accelerometer inputs. $g$ is the gravity. Several landmarks are observed as
\begin{equation}
    y_t^{(i)}=R_t^{-1}(d_i-p_t),\ 1\le i\le N.
\end{equation}
$d_i$ are known positions of the landmarks.

To construct an LTV embedding, we begin by letting $z_0^{(i)}=R_t^{-1}(d_i-p_t)$. Taking the derivative, we obtain 
\begin{equation}
    \dot{z}_0^{(i)}=-(R_t^{-1}\dot{R}_t)z_0^{(i)}-R_t^{-1}\dot{p}=-\omega_t^\times z_0^{(i)}-R_t^{-1}v_t.
\end{equation}
Since $R_t^{-1}v_t$ does not appear in previously embedded states, let $z_1=R_t^{-1}v_t$ and take derivative along system trajectory as 
\begin{equation}
    \begin{aligned}
    \dot{z}_1&=-(R_t^{-1}\dot{R}_t) z_1+R_t^{-1}(R_ta_t+g)\\
    &=-\omega_t^\times z_1+R_t^{-1}g+a_t
    \end{aligned}.
\end{equation}
Similarly, let $z_2=-R_t^{-1}g$, we have
\begin{equation}
    \dot{z}_2=-(R_t^{-1}\dot{R}_t)R_t^{-1}g=-\omega_t^\times z_2.
\end{equation}
Interestingly, no new term emerges that does not belong to the previous $z$. We obtain finite termination and construct an embedding as
\begin{equation}
    \label{eq::imu_lmk_ltv}
    \left\{\begin{aligned}
        \dot{z}_0^{(i)}&=-\omega_t^\times z_0^{(i)}-z_1\\
        \dot{z}_1&=-\omega_t^\times z_1-z_2+a_t\\
        \dot{z}_2&=-\omega_t^\times z_2\\
        y_t^{(i)}&=z_0^{(i)}
    \end{aligned}\right.,\ 1\le i\le N.
\end{equation}
Finite termination is not a coincidence, but the natural consequence of the linear observed structure, which we will reveal later. Again, the LTV system \eqref{eq::imu_lmk_ltv} is uniformly observable, and thus a Kalman observer for it \eqref{eq::imu_lmk_ltv} will converge uniformly exponentially. Leveraging the estimates of $z_i$, it remains to reconstruct the $\text{SE}_2(3)$-valued state $(\hat{R},\hat{p},\hat{v})$ through the following optimization
\begin{equation}
    \label{eq::imu_lmk_opt}
    \mathop{\min}_{R\in\SO(3),\ p\in\mathbb{R}^3,\ v\in\mathbb{R}^3}\left\Vert\hat{Z}-\begin{bmatrix}
        R & p & v\\
        0 & 1 & 0\\
        0 & 0 & 1
    \end{bmatrix}^{-1}D\right\Vert^2,
\end{equation}
where the coefficient matrices take the form
\begin{align}
    \hat{Z}&=\begin{bmatrix}
        \hat{z}_0^{(1)} & \hat{z}_0^{(2)} & \cdots & \hat{z}_0^{(N)} & \hat{z}_1 & \hat{z}_2\\
        1 & 1 & \cdots & 1 & 0 & 0\\
        0 & 0 & \cdots & 0 & -1 & 0
    \end{bmatrix}\in\mathbb{R}^{5\times(N+2)},\\
    D&=\begin{bmatrix}
        d_1 & d_2 & \cdots & d_N & 0_{3\times 1} & -g\\
        1 & 1 & \cdots & 1 & 0 & 0\\
        0 & 0 & \cdots & 0 & -1 & 0
    \end{bmatrix}\in\mathbb{R}^{5\times(N+2)}.
\end{align}
The systematic construction of $\hat{Z}$ and $D$ will be detailed in later sections. Moreover, the cost can be weighted for scaling reasons in practice. In this case, the weight is absorbed into $\hat{Z}$ and $D$. To proceed, we define the row block decomposition $\hat{Z}=[\hat{\bar{Z}}^\top,\hat{\underline{Z}}]^\top$ and $D=[\bar{D}^\top,\underline{D}^\top]^\top$, where $\hat{\bar{Z}},\bar{D}\in\mathbb{R}^{3\times(N+2)}$ and $\hat{\underline{Z}},\underline{D}\in\mathbb{R}^{2\times(N+2)}$, respectively. Let the SVD decomposition be
\begin{equation}
    \hat{\bar{Z}}\left[I_{5\times 5}-\underline{D}^\top\left(\underline{D}\underline{D}^\top\right)^{-1}\underline{D}\right]\bar{D}^\top=U\Lambda V^\top,
\end{equation}
where $U,V\in\text{O}(3)$ and $\Lambda\in\mathbb{R}^{3\times 3}$ is diagonal. $S=\text{diag}(1,1,\det(UV))$. Finally, we have the state estimation
\begin{align}
    \hat{R}&=VSU^\top,\\
    \label{eq::imu_lmk_pos_recon}
    [\hat{p},\hat{v}]&=(\bar{D}-VSU^\top\hat{\bar{Z}})\underline{D}^\top(\underline{D}\underline{D}^\top)^{-1}.
\end{align}

The success of the reconstruction, i.e., the uniqueness of the solution of \eqref{eq::imu_lmk_opt}, requires $\text{rank}(D)$ to be full. If there are at least three linearly independent landmarks, $\text{rank}(D)=5$. At the same time, an invariant filter for this problem also requires at least three non-colinear landmarks to make the system observable in logarithmic coordinates of $\text{SE}_2(3)$, hence the extrinsic and intrinsic observability conditions coincide.

\section{LTV Embedding of a Class of Linear Observed Systems on Groups}
\label{sec::properties_allowing_ltv_embedding}

Let $G\subset\mathbb{R}^{d_y\times d_y}$ be a matrix Lie group. Let $u\in\mathbb{R}^{\dim\mathfrak{g}}$ be time-dependent and $y^{(i)}\in\mathbb{R}^{d_y},\ 1\le i\le M$ be the input and outputs, respectively. A linear observed system on $G$ is given by
\begin{equation}
    \label{eq::los_system}
    \dot{\chi}=f_u(\chi),\quad y^{(i)}=h_i(\chi),\ i=1,2,...,M,
\end{equation}
where $f_u$ is a smooth group-affine vector field on $G$, and $h_i$ is a left- or right-linear action. The group-affine property is a structural condition \cite{InEKF} satisfying
\begin{equation}
    \label{eq::property_group_affine}
    f_u(\chi_1\chi_2)=\chi_1f_u(\chi_2)+f_u(\chi_1)\chi_2-\chi_1f_u(id_G)\chi_2,
\end{equation}
for all $\chi_1,\chi_2\in G$ and all $u\in\mathbb{R}^{\dim\mathfrak{g}}$. $h_i$ is defined by the natural linear action of $G$ on a known vector $d^{(i)}\in\mathbb{R}^{d_y}$ as
\begin{equation}
    h_i(\chi)=\chi^{-1}d^{(i)}\ \text{or}\ h_i(\chi)=\chi d^{(i)}.
\end{equation}

\subsection{A Class of Linear Observed Systems}

Using \eqref{eq::property_group_affine}, we have the decomposition $f(\chi)=g(\chi)+\chi f_u(id)$ or $f(\chi)=g(\chi)+f_u(id)\chi$ \cite{InEKF}. The vector field $g$ satisfies $g(\chi_1\chi_2)=g(\chi_1)\chi_2+\chi_1 g(\chi_2),\forall\chi_1,\chi_2\in G$. It is clear that $g$ is an infinitesimal generator of a one-parameter family of $\Aut(G)$, denoted by $\aut(G)$.

A special class of $g$ is related to the inner automorphisms $\Inn(G)\subset\Aut(G)$. Let $A\in\mathfrak{g}$ be identified with an $\mathbb{R}^{d_y\times d_y}$ matrix. This class of $g$ has the explicit expression $g_A(\chi)=A\chi-\chi A$. It is well known that the dynamics are bi-invariant when $g\in\inn(G)$. This compels us to consider group-affine vector fields $f_u$ corresponding to automorphisms beyond $\Inn(G)$, which are highly non-trivial. The quotient $\Out(G):=\Aut(G)/\Inn(G)$ is a well-defined abstract group. In general, it's impossible to obtain explicit formulas for $\Out(G)$. Luckily, such obstruction can be overcome through group embedding. We suppose that there exists another matrix Lie group $\tilde G\subset\GL(d_y,\mathbb{R})$ such that $G\subset\tilde G$ is a nontrivial normal subgroup of $\tilde G$. Let the Lie algebra of $\tilde G$ be $\tilde{\mathfrak{g}}$, which contains $\mathfrak{g}$. Consider the restriction of a vector field in $\inn(\tilde G)$ to $G$ as
\begin{equation*}
    g_{\tilde{A}}(\chi)=\tilde{A}\chi-\chi\tilde{A},\ \tilde{A}\in\tilde{\mathfrak{g}}.
\end{equation*}
If $\tilde A\in\tilde{\mathfrak{g}}$ but $\tilde A\not\in\mathfrak{g}$, then $g_{\tilde A}(\chi)$ is in $\aut(G)$ but not in $\inn(G)$. This means that $g_{\tilde{A}}$ corresponds to nontrivial outer automorphisms of $G$.

With the above notations, we now refine \eqref{eq::los_system} and define the class of systems to be studied throughout the paper:
\begin{align}
    \label{eq::def_type_1_sys}
    &\text{Type-1:}\ \ \left\{\begin{aligned}
        \dot{\chi}&=g_{\tilde{A}}(\chi)+\chi f_u(id)\\
        y^{(i)}&=\chi^{-1}d^{(i)},\ i=1,2,...,M
    \end{aligned}\right.,\\
    \label{eq::def_type_2_sys}
    &\text{Type-2:}\ \ \left\{\begin{aligned}
        \dot{\chi}&=g_{\tilde{A}}(\chi)+f_u(id)\chi\\
        y^{(i)}&=\chi d^{(i)},\ i=1,2,...,M
    \end{aligned}\right..
\end{align}
We emphasize that $\tilde{A}$ and all $d^{(i)}$ are time-independent. The input affects the dynamics only through a left- or right-invariant vector field. This is a true generalization of linear time-invariant systems to matrix Lie groups, as the drift term is time-independent and is related to an automorphism of the state space. Moreover, the control inputs are invariant, generalizing the constant vector fields in linear systems.

\begin{remark}
    In invariant filtering, defining a right-invariant error $e=\chi\hat{\chi}^{-1}$ for a Type-1 system or a left-invariant error $e=\hat{\chi}^{-1}\chi$ for a Type-2 system yields autonomous error dynamics $\dot{e}=g_{\tilde{A}}(e)$. Since $\tilde{A}$ and $d^{(i)}$ are constant, the Jacobians in logarithmic coordinates are independent of both the estimated state and time. 
\end{remark}

\subsection{LTV Embedding of the Class of System}

Unlike invariant \cite{InEKF,AnnuRevInEKF} or equivariant \cite{EqF, EqFCDC, Annurev_EqF} filters, which linearize the system in exponential (normal) coordinates and thus preclude global convergence, we avoid linearization through a powerful construction that allows a global embedding of the original system on the Lie group into an auxiliary high-dimensional linear system by delving deeper into the fine structure of the automorphism group. It is worth emphasizing that, unlike system immersion in \cite{nonlinear_immersion2}, which requires rank observability, our embedding mechanism is fundamentally different and is a consequence of the linear observed structure.

System \eqref{eq::los_system} embeds into an LTV system if there exists a smooth map $\pi:G\rightarrow\mathbb{R}^{d_z}$ such that for any common input and any pair of initial values aligned by $\pi$, i.e., $z_0=\pi(\chi_0)$, the outputs of the two systems coincide for all subsequent times \cite{nonlinear_immersion1, nonlinear_immersion2}. In contrast to LTV immersion, which is a local object valid in some open neighborhood, system embedding is a global object valid on the entire state space.

The existence of an embedded system is based on a finite termination criterion similar to that in Section~\ref{subsec::embedding_lti}.

\begin{theorem}\label{theorem::embedding_ltv}
    A Type-1 or Type-2 system defined by \eqref{eq::def_type_1_sys} or \eqref{eq::def_type_2_sys} can be embedded into an LTV system. The embedding map $\pi$ and the LTV system are given below. For simplicity, we consider only one measurement and omit the superscript $i$ in $y^{(i)}$ and $d^{(i)}$.
    \begin{itemize}
        \item For a Type-1 system, the embedding map is $\pi(\chi)=z:=[z_0^\top,\dots,z_{d_y-1}^\top]^\top\in\mathbb{R}^{d_y^2}$, whose row blocks are
        \begin{equation}
            \label{eq::pi_type_1}
            z_j:=\chi^{-1}(\tilde{A}^jd)\in\mathbb{R}^{d_y},\ 0\le j\le d_y-1,
        \end{equation}
        where $\tilde{A}^j$ denotes the $j$-th power of $\tilde{A}$. There exist constants $a_l\in\mathbb{R},\ l=0,1,...,d_y-1$, such that the dynamics governing $z$ are
        \begin{equation}
            \label{eq::ltv_type_1}
            \left\{\begin{aligned}
            &\dot z_j=-S_uz_j-z_{j+1},\ j=0,1,...,d_y-2\\
            &\dot z_{d_y-1}=-S_uz_{d_y-1}-\sum_{l=0}^{d_y-1}a_lz_l 
            \end{aligned}\right.,
        \end{equation}
        where $S_u:=f_u(id)-\tilde{A}$. The measurement is $y=z_0$.
        \item For a Type-2 system, the embedding map is $\pi(\chi)=z:=[z_0^\top,\dots,z_{d_y-1}^\top]^\top\in\mathbb{R}^{d_y^2}$, whose row blocks are
        \begin{equation}
            \label{eq::pi_type_2}
            z_j:=\chi(\tilde{A}^jd)\in\mathbb{R}^{d_y},\ 0\le j\le d_y-1.
        \end{equation}
        The same constants $a_l$ as in the Type-1 case are used. The dynamics governing $z$ are the same as \eqref{eq::ltv_type_1}, except that $S_u:=-f_u(id)-\tilde{A}$. The measurement is $y=z_0$.
    \end{itemize}
\end{theorem}
\begin{proof}
    See Appendix~\ref{app::proof_embedding_ltv}.
\end{proof}

\begin{remark}
    To guarantee observability, we must consider a system with multiple measurements, e.g., $y^{(i)}=\chi^{-1}d^{(i)}\in\mathbb{R}^{d_y},i=1,\dots,M$. The embedding is then performed $M$ times for each measurement, creating $M$ copies of embedded states $z^{(i)}:=[z_0^{(i)\top},\dots,z_{d_y-1}^{(i)\top}]^\top\in\mathbb{R}^{d_y^2}$. Stacking these states $z^{(i)}$ together yields an embedded LTV system on $\mathbb{R}^{Md_y^2}$. For a Type-1 system, the embedded system is given by
    \begin{equation}
        \label{eq::embedded_LTV}
        \left\{\begin{aligned}
        &\dot{z}^{(i)}_j=-S_u z^{(i)}_j-z^{(i)}_{j+1},\ 0\le j\le d_y-2\\
        &\dot z^{(i)}_{d_y-1}=-S_uz^{(i)}_{d_y-1}-\sum_{l=0}^{d_y-1}\tilde{a}_lz^{(i)}_l
        \end{aligned}\right.,
    \end{equation}
    with $M$ multiple measurements in $\mathbb{R}^{d_y}$ as $y^{(i)}=z^{(i)}_0,\ i=1,...,M$. For a Type-2 system, the LTV is the same as above with a different $S_u$.
\end{remark}

\begin{remark}
    In practice, one often considers joint estimation of an unknown constant input bias $b\in\mathbb{R}^{\dim\mathfrak{g}}$. Incorporating $b$ into the extended state in $G\times\mathfrak{g}$ destroys the group-affine property and leads to imperfect IEKFs \cite{TFG}. Despite efforts to modify the filter errors \cite{TFG, EqFBias,af_thesis}, local stability has not been proved with unknown input bias. Nevertheless, our embedding approach remains applicable with input bias $b$ by adding an additional equation $\dot{b}=0$ and replacing $f_u(id)$ with $f_{u+b}(id)$ in the embedded LTV system, shedding light on this issue. The equations governing the extended state $(z,b)$ which evolves in the vector space $\mathbb{R}^{Md_y^2}\times\mathbb{R}^{\dim\mathfrak{g}}$, become nonlinear.
\end{remark}

\section{Global Observer Design Using LTV embedding}
\label{sec::global_observer_using_ltv_embedding}

We focus on embeddable linear observed systems on $G$ given by \eqref{eq::def_type_1_sys} or \eqref{eq::def_type_2_sys} with $M$ measurements.
A global observer is constructed by first designing an observer for the embedded LTV and then reconstructing the state in the Lie group using the estimate of the LTV state. Note that our system model incorporates multiple measurements to ensure observability, similar to the formulation in \cite{InEKF}. 

\subsection{Observer Structure from the Embedding}

An observer for \eqref{eq::def_type_1_sys} or \eqref{eq::def_type_2_sys} is designed in two steps:
\begin{enumerate}
    \renewcommand{\labelenumi}{(\arabic{enumi})}
    \item implement a Kalman filter for the embedded LTV;
    \item use the estimate to reconstruct the group-valued state via optimization.
\end{enumerate}

Implementation of a Kalman filter for \eqref{eq::embedded_LTV} following Section~\ref{subsec::ekf} is straightforward. Suppose now we have obtained the estimate $\hat{z}$ from the first step, directly pulling back the linear observer to a dynamical system on $G$ is problematic, and thus the reconstruction of the state $\hat\chi\in G$ is formulated as solving an optimization problem \cite{obs_design_nls}. The reconstruction map $\hat{\chi}=\mathcal{T}(\hat{z})$ is defined as a left inverse of the embedding $\pi$ by
\begin{equation}
    \mathcal{T}:\mathbb{R}^{d_z}\rightarrow G,\ \hat{z}\mapsto\mathop{\arg\min}_{\hat{\chi}\in G}\Vert \hat{z}-\pi(\hat{\chi})\Vert^2.
\end{equation} 
For Type-1 systems, we have $z_j^{(i)}=\chi^{-1}(\tilde{A}^jd^{(i)})$. Let $d^{(i)}_j:=\tilde A^jd^{(i)}\in\mathbb{R}^{d_y}$. Hence, the reconstruction process is explicitly reformulated as
\begin{align}
    \label{eq::general_matrix_optimization_right}
    \hat{\chi}=\mathcal{T}(\hat{z})=\mathop{\arg\min}_{\hat{\chi}\in G}\left\Vert\hat{Z}-\hat{\chi}^{-1}D\right\Vert_{\Sigma}^2,
\end{align}
with the matrices $\hat{Z}, D\in\mathbb{R}^{d_y\times Md_y}$ given by $\hat{Z}:=\left[\hat{z}^{(1)}_0,...,\hat{z}^{(1)}_{d_y-1},...,\hat{z}^{(i)}_j,...,\hat{z}^{(M)}_0,...,\hat{z}^{(M)}_{d_y-1}\right]$ and $D:=\left[d^{(1)}_0,...,d^{(1)}_{d_y-1},...,d^{(i)}_j,...,d^{(M)}_0,...,d^{(M)}_{d_y-1}\right]$. Each column of $\hat{Z}$ is an estimate of the $j$-th component of the state of the embedded LTV corresponding to the $i$-th output, and each column of $D$ is the system structure. $d_0^{(i)}$ is directly obtained from the measurement equation, while $d_j^{(i)}$ results from the action of the powers of $\tilde{A}$ on $d_0^{(i)}$. Note that the index $i$ ranges from $[1,M]$ and $j$ ranges from $[0,d_y-1]$. The optimization problem \eqref{eq::general_matrix_optimization_right} can be weighted by a positive definite matrix $\Sigma\in\mathbb{S}^{Md_y}_+$.

For Type-2 systems, we have $z_j^{(i)}=\chi(\tilde{A}^jd^{(i)})$. $\hat{Z}$ and $D$ are defined as before. The optimization formulation is slightly different from \eqref{eq::general_matrix_optimization_right}:
\begin{equation}
    \label{eq::general_matrix_optimization_left}
    \hat{\chi}=\mathcal{T}(\hat{z})=\mathop{\arg\min}_{\hat{\chi}\in G}\left\Vert\hat{Z}-\hat{\chi}D\right\Vert_{\Sigma}^2.
\end{equation}

Since the constant weighting matrix $\Sigma$ can be absorbed into the norm by redefining $\hat{Z}$ and $D$ as $\hat{Z}\Sigma^{-\frac{1}{2}}$ and $D\Sigma^{-\frac{1}{2}}$, respectively, we omit $\Sigma$ for notational simplicity.

\begin{remark}
    Solving the optimization problem defined in \eqref{eq::general_matrix_optimization_right} or \eqref{eq::general_matrix_optimization_left} is nontrivial. The cost function is generally nonconvex and has multiple critical points due to the non-Euclidean topology of the state space $G$ \cite{Morse_theory}. Fortunately, for two-frame systems, the global minimum can be found explicitly through generalizing the Umeyama techniques in \cite{umeyama_algorithm}. If the global minimum is unique, any algorithm capable of escaping a saddle point or local maximum will converge to this global minimum. Hence, we do not specify a particular optimization algorithm for reconstructing $\hat{\chi}$. The optimization is expected to be solved at each time step when the Kalman observer provides an estimate. For two-frame systems, the computational cost is equivalent to that of an SVD. In resource-constrained environments, a gradient-like observer may be implemented to perform a single descent step at each time step. However, the algorithm must be carefully designed to avoid saddle points, which will inevitably appear.
\end{remark}

\subsection{Joint Estimation of Lie-Algebra-Valued Input Biases}\label{subsec::obs_with_bias}

If the input of the system on $G$ is corrupted by a constant $\mathfrak{g}$-valued bias, the embedding into $\mathbb{R}^{d_z}$ is preserved, and such a bias can be simultaneously estimated through an observer for the embedded system with an augmented bias state. Let $d_b:=\dim\mathfrak{g}$. To address bias estimation, we characterize the dependence of $f_u(id)$ on the input bias $b\in\mathbb{R}^{d_b}$.

\begin{assumption}\label{assumption::additivity_bias}
    $f_{u-b}(id)$ is affine in the input bias $b\in\mathbb{R}^{d_b}$, i.e., $f_{u-b}(id)=f_u(id)-\mathcal{L}_{\mathfrak{g}}(b)$.
\end{assumption}

Thus, $S_{u-b}=S_u-\mathcal{L}_{\mathfrak{g}}(b)$ for Type-1 systems and $S_{u-b}=S_u+\mathcal{L}_{\mathfrak{g}}(b)$ for Type-2 systems. The embedded system with bias is given by 
\begin{equation}
\label{eq::embedded_system_bias}
    \left\{\begin{aligned}
    &\dot{z}^{(i)}_j=-(S_u\mp\mathcal{L}_{\mathfrak{g}}(b)) z^{(i)}_j-z^{(i)}_{j+1}\\
    &\dot z^{(i)}_{d_y-1}=-(S_u\mp\mathcal{L}_{\mathfrak{g}}(b))z^{(i)}_{d_y-1}-\sum_{l=0}^{d_y-1}\tilde{a}_lz^{(i)}_l\\
    &\dot{b}=0,\quad y^{(i)}=z_0^{(i)},\ 1\le i\le M,\ 0\le j\le d_y-2
    \end{aligned}\right.,
\end{equation}
The bias state introduces multiplicative nonlinearities to the embedded system from the coupling of the bias and the original states $z^{(i)}_j$. An error-state Kalman-like observer for \eqref{eq::embedded_system_bias} is implemented following Section~\ref{subsec::ekf}. Using the estimates of $\hat{z}_j^{(i)}$, the group-valued state is reconstructed as before.

\subsection{Global Properties of the Bias-free Observer}

The proposed observer, consisting of a Kalman observer followed by an optimization, exhibits global stability if the optimization algorithm attains a global minimum under certain regularity conditions and the Kalman observer is globally exponentially stable. 

\begin{assumption}\label{assumption::lower_boundedness_transition}
    There exist constants $\alpha_1,\alpha_2>0$ such that $\alpha_1 I\preceq\psi^\top(t)\psi(t)\preceq\alpha_2 I,\forall t$, where $\dot{\psi}(t)=-S_u\psi(t)$ and $S_u$ is given in Theorem~\ref{theorem::embedding_ltv} for Type-1 and Type-2 systems.
\end{assumption}

\begin{lemma}\label{lm::FH_observability}
    Under Assumption~\ref{assumption::lower_boundedness_transition}, the embedded LTV system is uniformly observable.
\end{lemma}

\begin{proof}
    See Appendix~\ref{app::proof_FH_observability}.
\end{proof}

\begin{remark}
    Many practical systems, e.g., two-frame systems \cite{TFG}, satisfy Assumption~\ref{assumption::lower_boundedness_transition} due to state group structure as shown later. The assumption is considered very weak.
\end{remark}

\begin{assumption}\label{assumption::structure_rank_condition}
    The matrix $D\in\mathbb{R}^{d_y\times Md_y}$ from the structure of the system satisfies $\text{rank}(D)=d_y$. Note that $d_y$ is the dimension of the square matrix into which $G$ is embedded.
\end{assumption}

\begin{remark}
    Since $G$ consists of invertible matrices, $Z=\chi^{-1}D$ or $Z=\chi D$ has a unique solution for $\chi$ determined by $Z$ and $D$ if $\text{rank}(D)=d_y$. This implies that the cost function achieves a unique global minimum of zero. Hence, the true $\chi$ can be uniquely determined from the true embedded system state. We refer to this rank condition as an extrinsic observability requirement for embedding-based observers.
\end{remark}

Let $\hat{\chi}$ and $\chi$ denote the estimated and true states on $G$, respectively. Since $G$ is identified with a closed subgroup of $\text{GL}(d_y,\mathbb{R})$, an extrinsic metric on $G$ is defined by $d(\hat{\chi},\chi)=\Vert \chi^{-1}-\hat{\chi}^{-1}\Vert$ or $d(\hat{\chi},\chi)=\Vert\chi-\hat{\chi}\Vert$. The former metric is used in the stability analysis for Type-1 systems, while the latter is used for Type-2 systems.

\begin{theorem}\label{theorem::stability_non_bias}
    Under Assumption~\ref{assumption::lower_boundedness_transition} and \ref{assumption::structure_rank_condition}, the proposed embedding-based observer for \eqref{eq::def_type_1_sys} or \eqref{eq::def_type_2_sys} is globally exponentially stable with respect to the corresponding metric defined above.
\end{theorem}

\begin{proof}
    See Appendix~\ref{app::proof_stability_non_bias}.
\end{proof}

\begin{remark}
    For the two-frame systems discussed later in Section~\ref{sec::two_frame_systems}, an explicit solution of a global minimum for \eqref{eq::general_matrix_optimization_right} or \eqref{eq::general_matrix_optimization_left} can be obtained, thereby enabling global observer realizations as per Theorem~\ref{theorem::stability_non_bias}.
\end{remark}

\subsection{Comparison with Invariant Filtering}

It is interesting to compare the observability assumption for stability in our embedding-based observer and in invariant filtering. We consider only Type-1 systems without loss of generality. Let $\mathcal{O}_1=D=[\tilde{A}^jd^{(i)}]\in\mathbb{R}^{d_y\times Md_y}$ with $1\le j\le d_y$ and $1\le i\le M$ be the extrinsic observability matrix required for the reconstruction from the embedded state to the original state and let $\mathcal{O}_2$ be the intrinsic observability matrix encountered in invariant filtering. Using the Lie algebra isomorphism $\mathcal{L}$, we define the right-invariant error $\chi\hat{\chi}^{-1}=\exp(\mathcal{L}(x))$, where $x\in\mathbb{R}^{\dim\mathfrak{g}}$. The error-state dynamics in logarithmic coordinates are $\dot{x}=\mathcal{L}^{-1}(\tilde{A}\mathcal{L}(x)-\mathcal{L}(x)\tilde{A})=[\mathcal{L}^{-1}(\tilde{A}),x]=ad_{\mathcal{L}^{-1}(\tilde{A})}x$, where $[\cdot,\cdot]$ denotes the Lie bracket on $\tilde{\mathfrak{g}}$ and $ad_{y}x:=[y,x]$ for $x,y\in\tilde{\mathfrak{g}}$. Since the innovation $\Delta^{(i)}$ in InEKF takes the form $\Delta^{(i)}=\hat{\chi}y^{(i)}-d^{(i)}=\exp^{-1}(\mathcal{L}(x))d^{(i)}-d^{(i)}$ \cite{InEKF}, the Jacobian related to the $i$-th measurement is $H^{(i)}=\mathcal{L}^\dagger(d^{(i)})$, where the operator $\mathcal{L}^\dagger\in\mathbb{R}^{d_y\times\dim\mathfrak{g}}$ satisfies $\mathcal{L}^{\dagger}(d^{(i)})x=-\mathcal{L}(x)d^{(i)}$. It is clear that the Jacobians are independent of time, hence the Kalman observability matrix in the InEKF is given by
\begin{equation}
    \mathcal{O}_2=\begin{bmatrix}
        \mathcal{O}_2^{(1)}\\
        \mathcal{O}_2^{(2)}\\
        \vdots\\
        \mathcal{O}_2^{(M)}
    \end{bmatrix},\ \mathcal{O}_2^{(i)}=\begin{bmatrix}
        \mathcal{L}^{\dagger}(d^{(i)})\\
        \mathcal{L}^{\dagger}(d^{(i)})ad_{\mathcal{L}^{-1}(\tilde{A})}\\
        \vdots\\
        \mathcal{L}^{\dagger}(d^{(i)})ad^{\dim\mathfrak{g}-1}_{\mathcal{L}^{-1}(\tilde{A})}\\
    \end{bmatrix}.
\end{equation} 

\begin{theorem}\label{th::obs_comparison}
    If ${O}_1$ has full rank, then $\mathcal{O}_2$ has full rank. Conversely, if for every subspace $S\neq\mathbb{R}^{d_y}$ of $\mathbb{R}^{d_y}$ invariant under $\tilde{A}$, i.e., $\tilde{A}S\subset S$, there exists a nonzero $M\in\mathfrak{g}$ such that $MS=0$, then $\mathcal{O}_2$ having full rank implies that $\mathcal{O}_1$ has full rank. 
\end{theorem}
\begin{proof}
    See Appendix~\ref{app::proof_obs_comparison}.
\end{proof}

The observability assumption required by the embedding-based design is sometimes stricter than that of the InEKF, a price paid for achieving global stability. In Section~\ref{subsec::sim_att_obs}, we remedy this issue by generating a third linear-independent landmark. In Section~\ref{subsec::imu_lmk_observer}, the observability requirements for both methods are the same.

\subsection{Semi-global Properties of the Observer with Bias}

With biases, the embedding approach remains applicable, but the resulting system is no longer linear, and the Jacobians of the extended Kalman filter depend on the estimated state $\hat{z}$. The stability of the extended Kalman filter applied to this system relies on the uniform boundedness of the eigenvalues of $P$, which depends on the uniform observability, and in turn, on the estimated state $\hat{z}$ through its Jacobians. The most challenging part is establishing the conditions for the uniform observability of the embedded system, after which the method in \cite{semi_global_ekf,ekf_modified_riccati_ode} provides a guarantee on the boundedness of $P$, leading to the following nonlocal results, as in \cite{semi_global_ekf}.

\begin{assumption}\label{assumption::bounded_true_trajectory}
    The true trajectory on the group evolves within a compact subset $\mathcal{G}_1$ of $G$.
\end{assumption}

\begin{lemma}\label{lm::uniform_observability}
    If the following conditions hold
    \begin{enumerate}
        \renewcommand{\labelenumi}{(\arabic{enumi})}
        \item there exist $\alpha_1,\alpha_2>0$ such that $\alpha_1I\preceq\psi^\top(t)\psi(t)\preceq\alpha_2I$ for all $t$ with $\dot{\psi}(t)=-S_u\psi(t)$;
        \item there exists $\mu_2>0$ such that $\mathcal{L}^\top(\hat{b}(t))\mathcal{L}(\hat{b}(t))\preceq\mu_2I$ for all $t$;
        \item there exist $\delta,\gamma_1,\beta>0$ such that for all $t$, $\Vert\Gamma(t)\Vert\le\beta$, and the excitation satisfies
        \begin{equation*}
            \frac{1}{\delta}\int_t^{t+\delta}\Gamma^\top(\tau)\Gamma(\tau)d\tau\succeq\gamma_1I,
        \end{equation*}
        where $\Gamma:=[\mathcal{L}^\ddagger(\hat{z}_0^{(i)})^\top, ..., \mathcal{L}^\ddagger(\hat{z}_j^{(i)})^\top, ..., \mathcal{L}^{\ddagger}(\hat{z}_{N-1}^{(i)})^\top]^\top$ for some $i$ with $1\le i\le M$, evaluated at the embedded estimate and $\mathcal{L}^\ddagger$ is the unique operator satisfying $\mathcal{L}^\ddagger(z)x=\mathcal{L}(x)z,\forall x\in\mathbb{R}^{\dim\mathfrak{g}},\forall z\in\mathbb{R}^{d_y}$;
        \item the excitation is sufficiently strong over a sufficiently small interval:
        \begin{equation*}
            \gamma_1e^{-8\delta c_\eta}>\frac{3\alpha_2^3\lambda^2_{\max}(R)\beta^2}{4\alpha_1^3\lambda_{\min}^2(R)},
        \end{equation*}
        where $c_\eta:=\frac{\alpha_2}{\alpha_1}\mu_2+\sqrt{\sum_{l=0}^{N-1}\tilde{a}_l^2+1}$. Note that $a_l$ are the coefficients from the system structure \eqref{eq::embedded_system_bias};
    \end{enumerate}
    then the Jacobian pair, $(\tilde{F}_u,\tilde{H})$ obtained by linearizing \eqref{eq::embedded_system_bias}, is uniformly observable and determinable.
\end{lemma}
\begin{proof}
    See Appendix~\ref{app:proof_uniform_observability}.
\end{proof}

\begin{theorem}\label{theorem::stability_bias}
    For the case with bias, consider the observer discussed in Section~\ref{subsec::obs_with_bias}. Under Assumptions~\ref{assumption::structure_rank_condition}, \ref{assumption::bounded_true_trajectory} and the conditions of Lemma~\ref{lm::uniform_observability}, i.e., sufficient excitation, there exist a compact subset $\mathcal{G}_2\subset G$ and a compact subset $\hat{\mathcal{B}}\subset\mathbb{R}^{\dim\mathfrak{g}}$ such that for any initial condition $\hat{\chi}(t_0)\in\text{int}(\mathcal{G}_1)$ and any $\hat{b}(t_0)\in\mathcal{B}\subset\hat{\mathcal{B}}$, the estimates $\hat{\chi}(t),\hat{b}(t)$ remain in $\mathcal{G}_2$ and $\hat{\mathcal{B}}$, respectively, for all $t\in[t_0,\infty)$. Moreover, $d(\hat{\chi}(t),\chi(t))$ and $\Vert\hat{b}(t)-b\Vert$ converge exponentially to zero after some finite time.
\end{theorem}

\begin{proof}
    See Appendix~\ref{app::proof_stability_bias}.
\end{proof}

Interestingly, the stability guarantee with bias involves assumptions on the strength of the  excitation rather than additional observability beyond $\text{rank}(D)$, such as additional landmarks $d^{(i)}$.

\section{Application to Two-Frame Systems}\label{sec::two_frame_systems}

We apply the proposed observer design toolbox to two-frame systems, which are linear observed systems on two-frame groups, constructed via the semidirect product of a rotation group with several vector spaces. Two-frame systems provide a powerful framework for modeling a broad class of navigation problems. Our theory provides unified GES observer designs, compared with the local results achieved by the InEKF \cite{TFG, InEKF} and case-by-case nonlinear constructive methods \cite{hybrid_observer_landmark} achieving global results. Without loss of generality, we consider only Type-1 embeddable systems.

\subsection{Embeddable Two-Frame Systems}

The two-frame group, denoted by $\TFG(d,n,m)$ \cite{TFG}, where $d=2$ or $3$ and $n,m\in\mathbb{N}$, is a matrix Lie group defined as the closed subgroup of $\GL(d+n+m,\mathbb{R})$ in the form of
\begin{equation*}
    \TFG(d,n,m)=\left\{\left.\begin{bmatrix}
        R & W\\
        0 & I\\
    \end{bmatrix}\right|\begin{aligned}
        &R\in\SO(d),W=[X\ RY]\\
        &X\in\mathbb{R}^{d\times n},Y\in\mathbb{R}^{d\times m}
    \end{aligned}    
    \right\}.
\end{equation*} 

The size of $W$ is $\mathbb{R}^{d\times(n+m)}$. Each column of $X$ and $Y$ is an $\mathbb{R}^d$-valued vector. The two-frame group describes the rigid geometric transformation between two frames, serving as the state space in single rigid-body kinematics. In navigation problems, suppose that $R$ represents the rotation from the body frame to the world frame, the $n$ columns of $X$ are related to $\mathbb{R}^d$-valued states expressed in the world frame, whereas the $m$ columns of $Y$ expressed in the body frame. In the embedding-based observer design, $W=[X,RY]$ is treated as a whole, and hence there is no fundamental difference between designing observers for TFG and $\text{SE}_{n+m}(d)$. We still use $\TFG(d,n,m)$, as it models a broader class of systems. Its Lie algebra $\tfg(d,n,m)\subset\mathbb{R}^{(d+n+m)\times(d+n+m)}$ is given by
\begin{equation*}
    \tfg(d,n,m)=\left\{\left.\begin{bmatrix}
        \omega^\times & \rho\\
        0 & 0
    \end{bmatrix}\right|\omega\in\mathbb{R}^{\frac{d(d-1)}{2}},\rho\in\mathbb{R}^{d\times(n+m)}   
    \right\}.
\end{equation*}
The extended similarity transformation group  $\SIM_{n+m}(d)$ \cite{synchronous_observer_design}, where $d=2$ or $3$ and $n,m\in\mathbb{N}$, is a matrix Lie group, defined as the closed subgroup of $\GL(d+n+m,\mathbb{R})$ as
\begin{equation*}
    \SIM_{n+m}(d)=\left\{\left.\begin{bmatrix}
        R & W\\
        0 & A\\
    \end{bmatrix}\right|\begin{aligned}
        &R\in\SO(d),W\in\mathbb{R}^{d\times(n+m)}\\
        &A\in\GL(n+m,\mathbb{R})
    \end{aligned}\right\}.
\end{equation*}

Its Lie algebra $\mathfrak{sim}_{n+m}(d)$ is given by
\begin{equation*}
    \mathfrak{sim}_{n+m}(d)=\left\{\left.\begin{bmatrix}
        \Omega^\times & \gamma\\
        0 & L
    \end{bmatrix}\right|\begin{aligned}    
        &\Omega\in\mathbb{R}^{\frac{d(d-1)}{2}},\gamma\in\mathbb{R}^{d\times(n+m)}\\
        & L\in\mathbb{R}^{(n+m)\times (n+m)}
    \end{aligned}   
    \right\}.
\end{equation*}

$\TFG(d,n,m)$ (or $\text{SE}_{n+m}(d)$) is a normal subgroup of $\SIM_{n+m}(d)$ \cite{synchronous_observer_design}. Using $\tfg$, we provide an explicit characterization of Type-1 systems following \eqref{eq::def_type_1_sys}. Let the state to be estimated $T\in\TFG(d,n,m)$ with block components as 
\begin{align}
    \label{eq::tfg_state}
    T=\begin{bmatrix}
        R & W\\
        0 & I
    \end{bmatrix}\in\TFG(d,n,m).
\end{align}

Letting $\Omega$, $\gamma$, $\tilde{\omega}$, and $\tilde{\rho}$ be matrix blocks of appropriate dimensions, we write
\begin{equation}
    \label{eq::tfg_case1_dynamics_complicated}
    \dot{T}_t=\begin{bmatrix}
        \Omega^\times & \gamma\\
        0 & L
    \end{bmatrix}T_t-T_t\begin{bmatrix}
        \Omega^\times & \gamma\\
        0 & L
    \end{bmatrix}+T_t\begin{bmatrix}
        \tilde{\omega}_t^\times & \tilde{\rho}_t\\
        0 & 0
    \end{bmatrix},
\end{equation}
following \eqref{eq::def_type_1_sys}, where we explicitly indicate the time dependence. For simplicity, we combine terms and define $\omega_t:=\tilde{\omega}_t-\Omega\in\mathbb{R}^{\frac{d(d-1)}{2}}$ and $\rho_t:=\tilde{\rho}_t-\gamma\in\mathbb{R}^{d\times(n+m)}$ as inputs. This yields the characterization of Type-1 embeddable two-frame systems as 
\begin{equation}
    \label{eq::tfg_case1_systems}
    \left\{\begin{aligned}\dot{T}_t&=\begin{bmatrix}
        \Omega^\times & \gamma\\
        0 & L
    \end{bmatrix}T_t+T_t\begin{bmatrix}
        \omega_t^\times & \rho_t\\
        0 & -L
    \end{bmatrix}\\
    y^{(i)}&=T_t^{-1}d^{(i)},\ i=1,2,...,M
    \end{aligned}\right.,
\end{equation}
where $d^{(i)}\in\mathbb{R}^{d+n+m},\ i=1,...,M$, are constant vectors. Note that $y^{(i)}\in\mathbb{R}^{d+n+m}$. Let the $\tfg$-matrix be
\begin{equation}
    \label{eq::tfg_aut}
    \tilde{A}=\begin{bmatrix}
        \Omega^\times & \gamma\\
        0 & L
    \end{bmatrix}\in\mathfrak{sim}_{n+m}(d).
\end{equation}

Let $N=d+n+m$. Let $\pi:\TFG(d,n,m)\rightarrow\mathbb{R}^{d_z}$ denote the embedding map for \eqref{eq::tfg_case1_systems}, i.e., $\pi(T)=z:=[z_0^{(1)\top},\dots,z_{N-1}^{(M)\top}]^\top\in\mathbb{R}^{d_z}$ with row blocks defined by $z_j^{(i)}=T^{-1}\tilde{A}^jd^{(i)},0\le j\le N-1,1\le i\le M$. Thus, $d_z=MN^2$, with $z_{j}^{(i)}\in\mathbb{R}^{N}$. The embedding of \eqref{eq::tfg_case1_systems} is
\begin{equation}
    \label{eq::immersion_tfg_right}
    \left\{\begin{aligned}
    \dot{z}^{(i)}_j&=-\begin{bmatrix}
        \omega_t^\times & \rho_t\\
        0 & -L
    \end{bmatrix}z^{(i)}_j-z^{(i)}_{j+1},\ j\in[0,N-2]\\
    \dot{z}^{(i)}_{N-1}&=-\begin{bmatrix}
        \omega_t^\times & \rho_t\\
        0 & -L
    \end{bmatrix}z^{(i)}_{N-1}-\sum_{l=0}^{N-1}\tilde{a}_lz^{(i)}_l\\
    y^{(i)}&=z_0^{(i)},\ i\in[1,M]
    \end{aligned}\right.,
\end{equation}
where the coefficients $\tilde{a}_j$ are obtained from the operator equation $\tilde{A}^{N}=\sum_{l=0}^{N-1}\tilde{a}_l\tilde{A}^l$. The variables involved in \eqref{eq::immersion_tfg_right} are homogeneous coordinates, and thus only the first $d$ coordinates of $z_j^{(i)}$ or $y^{(i)}$ are of interest in practice. Decompose the state and measurement as $z_j^{(i)}=[\bar{z}_j^{(i)\top},\underline{z}_j^{(i)\top}]^\top$ and $y^{(i)}=[\bar{y}^{(i)\top},\underline{y}^{(i)\top}]^\top$, where $\bar{z}_j^{(i)},\bar{y}^{(i)}\in\mathbb{R}^d$ and $\underline{z}_j^{(i)},\underline{y}^{(i)}\in\mathbb{R}^{n+m}$. \eqref{eq::immersion_tfg_right} is divided into two subsystems
\begin{align}
    \label{eq::immersion_tfg_right_bar}
    &\left\{\begin{aligned}
        \dot{\bar{z}}^{(i)}_j&=-\omega_t^\times\bar{z}_j-\rho_t\underline{z}_j^{(i)}-\bar{z}_{j+1}^{(i)},\ j\in[0,N-2]\\
        \dot{\bar{z}}_{N-1}^{(i)}&=-\omega_t^\times\bar{z}_{N-1}^{(i)}-\rho_t\underline{z}_{N-1}^{(i)}-\sum_{l=0}^{N-1}\tilde a_l\bar{z}_l^{(i)}\\
        \bar{y}^{(i)}&=\bar{z}_0^{(i)},\ i\in[1,M]
    \end{aligned}\right.,\\
    \label{eq::immersion_tfg_right_underline}
    &\left\{\begin{aligned}
        \dot{\underline{z}}_j^{(i)}&=L\underline{z}_j^{(i)}-\underline{z}_{j+1}^{(i)},\ j\in[0,N-2]\\
        \dot{\underline{z}}_{N-1}^{(i)}&=L\underline{z}_{N-1}^{(i)}-\sum_{l=0}^{N-1}\tilde{a}_l\underline{z}_l^{(i)}\\
        \underline{y}^{(i)}&=\underline{z}_0^{(i)},\ i\in[1,M]
    \end{aligned}\right.,
\end{align}
where \eqref{eq::immersion_tfg_right_bar} is cascaded with \eqref{eq::immersion_tfg_right_underline}. 

As before, we define the notation $d_j^{(i)}=\tilde{A}^jd^{(i)}$. Defining $d_j^{(i)}=[\bar{d}_j^{(i)\top},\underline{d}_j^{(i)\top}]^\top$, where $\bar{d}_j^{(i)}\in\mathbb{R}^d$ and $\underline{d}_j^{(i)}\in\mathbb{R}^{n+m}$, these components are calculated inductively by
\begin{align}
    \label{eq::dij_underline}
    \underline{d}^{(i)}_j&=L^j\underline{d}^{(i)},\ i\in[1,M],\ j\in[0,N-1],\\
    \label{eq::dij_bar}
    \bar{d}^{(i)}_{j+1}&=\Omega^\times\bar{d}^{(i)}_j+\gamma\underline{d}_j^{(i)},\ i\in[1,M],\ j\in[0,N-2].
\end{align}

Since all underlined variables result from extending the physical coordinates to homogeneous coordinates, it is known a priori that $\underline{y}^{(i)}=\underline{z}_0^{(i)}=\underline{d}^{(i)}_0=\underline{d}^{(i)}$ are constants. Let the initial values of \eqref{eq::immersion_tfg_right_underline} be $\underline{z}_j^{(i)}(t_0)=\underline{d}^{(i)}_j=L^j\underline{d}^{(i)}$. The subsystem \eqref{eq::immersion_tfg_right_underline} remains constant as the right-hand side of \eqref{eq::immersion_tfg_right_underline} vanishes for all $t\ge t_0$ by virtue of $L^{N}=\sum_{l=0}^{N-1}\tilde{a}_lL^l$, which follows from the definition of $\tilde{a}_l$. Hence, it suffices to design an observer for the subsystem \eqref{eq::immersion_tfg_right_bar}. Since $\underline{z}^{(i)}_j\equiv\underline{d}_j^{(i)}$, substituting $\underline{z}_j^{(i)}$ with $\underline{d}_j^{(i)}$, the embedded system \eqref{eq::immersion_tfg_right} becomes
\begin{align}
    \label{eq::immersion_tfg_right_bar_final}
    \left\{\begin{aligned}
        \dot{\bar{z}}^{(i)}_j&=-\omega_t^\times\bar{z}^{(i)}_j-\rho_t\underline{d}_j^{(i)}-\bar{z}_{j+1}^{(i)},\ j\in[0,N-2]\\
        \dot{\bar{z}}_{N-1}^{(i)}&=-\omega_t^\times\bar{z}_{N-1}^{(i)}-\rho_t\underline{d}_{N-1}^{(i)}-\sum_{l=0}^{N-1}\tilde a_l\bar{z}_l^{(i)}\\
        \bar{y}^{(i)}&=\bar{z}_0^{(i)},\ i\in[1,M]
    \end{aligned}\right..
\end{align}

A Kalman observer is designed for \eqref{eq::immersion_tfg_right_bar_final}. Let the observer state be $\hatbar{z}:=[\hatbar{z}_0^{(1)\top},...,\hatbar{z}_j^{(i)\top},...,\hatbar{z}_{N-1}^{(M)\top}]^\top\in\mathbb{R}^{MNd}$ with $0\le j\le N-1$ and $1\le i\le M$, where $\hatbar{z}_j^{(i)}$ is in $\mathbb{R}^d$. Let $\bar{y}=[\bar{y}^{(1)\top},...,\bar{y}^{(M)\top}]^\top\in\mathbb{R}^{Md}$ denote the stacked output. The observer dynamics are given by
\begin{align}
    \label{eq::tfg_immersion_observer}
    \dot{\hatbar{z}}=F_u\hatbar{z}+C_u+K(\bar{y}-H\hatbar{z}),
\end{align}
where the Kalman gain $K\in\mathbb{R}^{MNd\times Md}$ is calculated using the Riccati ODE. $H=\text{diag}(H^{(1)},...,H^{(M)})$ with blocks $H^{(i)}=[I_{d\times d}, 0_{d\times d(N-1)}], 1\le i\le M$. The $(F_u,C_u)$ pair is composed of blocks $F_u=\text{diag}(F_u^{(1)},\dots,F_u^{(M)})$ and $C_u=[C_u^{(1)\top},\dots,C_u^{(M)\top}]^\top$, where $F_u^{(i)}\in\mathbb{R}^{Nd\times Nd}$ and $C_u^{(i)}\in\mathbb{R}^{Nd}$ correspond to the $i$-th measurement. These blocks are given by \eqref{eq::immersion_tfg_right_bar_final} as
\begin{align*}
    F_u^{(i)}&=\begin{bmatrix}
        -\omega_t^\times & -I & 0 & \cdots & 0\\
        0 & -\omega_t^\times & -I & \cdots & 0\\
        \vdots & \vdots & \vdots & \vdots & \vdots\\
        -\tilde{a}_0I & -\tilde{a}_1I & -\tilde{a}_2I & \cdots & -\omega_t^\times-\tilde{a}_{N-1}I
    \end{bmatrix},\\
    C_u^{(i)}&=\begin{bmatrix}
        -\left(\rho_t\underline{d}_0^{(i)}\right)^\top & \cdots &
        -\left(\rho_t\underline{d}_{N-1}^{(i)}\right)^\top
    \end{bmatrix}^\top.
\end{align*}

\subsection{Two-frame Group State Reconstruction}

We now formulate the state reconstruction procedure \eqref{eq::general_matrix_optimization_right} or \eqref{eq::general_matrix_optimization_left} for two-frame systems. Let $\hat{Z}=[\hatbar{Z}^\top,\underline{Z}^\top]^\top$ and $D=[\bar{D}^\top,\underline{D}^\top]^\top$. Their components are obtained from the estimates of the embedded LTV system, as
\begin{align}
    \label{eq::bar_Z}
    \hatbar{Z}&:=\left[
        \hatbar{z}^{(1)}_0,\dots,\hatbar{z}^{(1)}_{N-1},\dots,\hatbar{z}^{(i)}_j,\dots,\hatbar{z}^{(M)}_0,\dots,\hatbar{z}^{(M)}_{N-1}
    \right],\\
    \label{eq::bar_D}
    \bar{D}&:=\left[
        \bar{d}^{(1)}_0,\dots,\bar{d}^{(1)}_{N-1},\dots,\bar{d}^{(i)}_j,\dots,\bar{d}^{(M)}_0,\dots,\bar{d}^{(M)}_{N-1}
    \right],\\
    \label{eq::underline_D}
    \underline{D}&=\left[
        \underline{d}^{(1)}_0,\dots,\underline{d}^{(1)}_{N-1},\dots,\underline{d}^{(i)}_j,\dots,\underline{d}^{(M)}_0,\dots,\underline{d}^{(M)}_{N-1}
    \right].
\end{align}
Moreover, $\underline{Z}=\underline{D}$. The dimensions of these matrix blocks are easily determined. Based on the Umeyama algorithm \cite{umeyama_algorithm}, we have the following lemma, which later solves \eqref{eq::general_matrix_optimization_right}.

\begin{lemma}\label{lm::umeyama}
    Assume that $\underline{D}\underline{D}^\top$ is invertible. Let a singular value decomposition be
    \begin{equation*}
        \bar{U}\Lambda\bar{V}^\top=\hatbar{Z}\left[I_{MN\times MN}-\underline{D}^\top(\underline{D}\underline{D}^\top)^{-1}\underline{D}\right]\bar{D}^\top 
    \end{equation*}
    with singular values $\Lambda:=\text{diag}(\sigma_1,...,\sigma_d)$ in decreasing order. Define $\bar{S}$ as
    \begin{align}
        \bar{S}=\begin{cases}
            I_{d\times d}, & \det(\bar{U}\bar{V})=1\\
            \text{diag}(I_{(d-1)\times(d-1)},-1), & \det(\bar{U}\bar{V})=-1
        \end{cases}.
    \end{align}
    
    Then, the global minimum of optimization
    \begin{align}
        \label{eq::umeyama_optimization}
        \mathop{\min}_{R\in\SO(d),W\in\mathbb{R}^{d\times(n+m)}}\left\Vert\begin{bmatrix}
            \hatbar{Z} \\ \underline{Z}
        \end{bmatrix}-\begin{bmatrix}
            R & W\\
            0 & I
        \end{bmatrix}^{-1}\begin{bmatrix}
            \bar{D} \\ \underline{D}
        \end{bmatrix}\right\Vert^2
    \end{align}
    is given by
    \begin{align}
        R^*&=\bar{V}\bar{S}\bar{U}^\top,\\
        W^*&=(\bar{D}-\bar{V}\bar{S}\bar{U}^\top\hatbar{Z})\underline{D}^\top(\underline{D}\underline{D}^\top)^{-1}.
    \end{align}
\end{lemma}

\begin{proof}
    See Appendix~\ref{app::proof_lm_umeyama}.
\end{proof}

We solve the TFG state reconstruction $\hat{\chi}=\mathcal{T}(\hat{z})$ formulated in \eqref{eq::general_matrix_optimization_right} using Lemma~\ref{lm::umeyama}. The constant weighting matrix $\Sigma$ is absorbed into the cost function by substituting $\hat{Z}$ and $D$ with $\hat{Z}\Sigma^{-\frac{1}{2}}$ and $D\Sigma^{-\frac{1}{2}}$. Define the state estimate $\hat{\chi}\in\TFG(d,n,m)$ as
\begin{align}
    \hat{\chi}=\begin{bmatrix}
        \hat{R} & \hat{W}\\
        0 & I
    \end{bmatrix}\in\TFG(d,n,m).
\end{align}
The state reconstruction for Type-1 systems using \eqref{eq::general_matrix_optimization_right} is
\begin{align}
    \label{eq::tfg_state_reconstruct_right}
    \left\{\begin{aligned}
        \hat{R}&=\bar{V}\bar{S}\bar{U}^\top\\
        [\hat{X}, \hat{R}\hat{Y}]&=\hat{W}=(\bar{D}-\bar{V}\bar{S}\bar{U}^\top\hatbar{Z})\underline{D}^\top(\underline{D}\underline{D}^\top)^{-1}
    \end{aligned}\right..
\end{align}

\subsection{Global Properties of the Bias-free Observer for Embeddable Two-frame Systems}

By Theorem~\ref{theorem::stability_non_bias}, the global stability of the proposed observer for two-frame systems is determined by the rank condition.

\begin{proposition}\label{proposition::tfg_stability_non_bias}
    The embedding-based observer for a Type-1 embeddable two-frame system is globally exponentially stable if $\text{rank}(D)=d+n+m$.
\end{proposition}

\begin{proof}
    It suffices to show that Assumption~\ref{assumption::lower_boundedness_transition} holds. For embeddable two-frame systems, $\dot{\psi}=-\omega_t^\times\psi$, which implies that $\psi(t)\in\SO(d)$ for all $t$, thereby completing the proof. 
\end{proof}

\subsection{Joint Estimation of Input Biases}

We consider joint estimation of a constant $\tfg(d,n,m)$-valued bias $b$, which involves only slight modifications to the embedded system on $\mathbb{R}^{d_z}$ and its corresponding estimator. The components of the bias are
\begin{align}
    \mathcal{L}_{\TFG}(b)=\begin{bmatrix}
        b_\omega^\times & b_\rho\\
        0 & 0
    \end{bmatrix}\in\tfg(d,n,m),
\end{align}
where $b_\omega\in\mathbb{R}^{\frac{d(d-1)}{2}}$ and $b_\rho\in\mathbb{R}^{d\times(n+m)}$. By Assumption~\ref{assumption::additivity_bias}, the embedded system involves the bias in an additive fashion, that is, we replace $(\omega_t,\rho_t)$ with $(\omega_t+b_\omega,\rho_t+b_\rho)$. For a Type-1 two-frame system, the biased embedded system is given by
\begin{align}
    \label{eq::immersion_bias_tfg_right_bar_final}
    \left\{\begin{aligned}
        \dot{\bar{z}}^{(i)}_j&=-\left(\omega_t+b_\omega\right)^\times\bar{z}^{(i)}_j-\left(\rho_t+b_\rho\right)\underline{d}_j^{(i)}-\bar{z}_{j+1}^{(i)}\\
        \dot{\bar{z}}_{N-1}^{(i)}&=-\left(\omega_t+b_\omega\right)^\times\bar{z}^{(i)}_{N-1}-\left(\rho_t+b_\rho\right)\underline{d}_{N-1}^{(i)}-\sum_{l=0}^{N-1}\tilde a_l\bar{z}_l^{(i)}\\
        \dot{b}_\omega&=0,\ \dot{b}_\rho=0,\quad j\in[0,N-2]\\
        \bar{y}^{(i)}&=\bar{z}_0^{(i)},\ i\in[1,M]
    \end{aligned}\right..
\end{align}

We vectorize $b_\rho$ by stacking its columns into a single vector.

\begin{proposition}\label{prop::vector_bias}
    If the embeddable two-frame system is subject only to $b_\rho$ associated with the vector dynamics, an embedding-based observer for the biased system \eqref{eq::immersion_bias_tfg_right_bar_final} is GES if $\text{rank}(D)=d+n+m$.
\end{proposition}
\begin{proof}
    See Appendix~\ref{app::proof_vector_bias}.
\end{proof}

\begin{proposition}
    If the embeddable two-frame system is subject to both $b_\rho$ and $b_\omega$, the embedding-based observer with bias is semi-globally stable, if (1) the conditions of Lemma~\ref{lm::uniform_observability} are satisfied; (2) the true trajectory is bounded; (3) $\text{rank}(D)=d+n+m$.
\end{proposition}

\begin{proof}
    This follows directly from Theorem~\ref{theorem::stability_bias}.
\end{proof}

The simultaneous estimation of angular bias makes the problem nontrivial, as indicated by Lemma~\ref{lm::uniform_observability}. The excitation of the angular part of the input must be sufficient large to satisfy the conditions of Lemma~\ref{lm::uniform_observability}, that is, $\Vert\omega_t\Vert$ or $\Vert\dot{\omega}_t\Vert$ should not be identically zero. In practice, one implements the observer and verifies that the smallest eigenvalue of $\mathcal{O}(t,t+\delta)$ is uniformly bounded below.

\subsection{Extension to Range and Bearing Measurements}

LTV embedding makes it possible to handle bearing and range measurements which are not compatible with the InEKF. Using the invariant error in these measurements introduces additional dependence of the innovation on the estimate, thereby destroying the local stability guarantee in the InEKF. 

Define the projection $\pi_{\mathbb{S}^n}:\mathbb{R}^{n+1}\rightarrow\mathbb{S}^n,x\mapsto x/\Vert x\Vert_2$. Bearing measurements for \eqref{eq::tfg_case1_systems} are given by
\begin{equation}
    y^{(i)}=\pi_{\mathbb{S}^{d-1}}\left(T_t^{-1}d^{(i)}\right).
\end{equation}
Thus, the original measurement equation $\bar{y}^{(i)}=\bar{z}_0^{(i)}$ in \eqref{eq::immersion_tfg_right_bar_final} is replaced with
\begin{equation}
    \label{eq::bearing_embedding}
    y^{(i)}=\pi_{\mathbb{S}^{d-1}}(\bar{z_0}^{(i)}),
\end{equation}
while the embedding remains valid. Although \eqref{eq::bearing_embedding} is nonlinear, it can be converted into a time-varying linear form using well-known orthogonal projection techniques widely used in \cite{hybrid_observer_vision_aided,riccati_observer_pnp, pose_velocity_landmark_position_wang, bearing_ins, bearing_ins_acc_bias,relative_pose_bearing} as
\begin{equation}
    \label{eq::bearing_linear_measurement}
    0=\left(I_{d\times d}-\bar{y}^{(i)}\bar{y}^{(i)\top}\right)\bar{z}_0^{(i)}:=\Pi_{\bar{y}^{(i)}}\bar{z}_0^{(i)},
\end{equation}
where the known trajectory of $\bar{y}^{(i)}$ (with $\Vert\bar{y}^{(i)}\Vert=1$) is injected into the linear measurement matrix, and $0$ is regarded as a virtual measurement.

For embeddable two-frame systems with bearing measurements, our observer is implemented in the same manner as in Proposition~\ref{proposition::tfg_stability_non_bias} with the difference that the measurement equations in the embedded system are replaced with \eqref{eq::bearing_linear_measurement}. In addition to the rank condition on $D$, global exponential stability is achieved if the Kalman observer for the embedded system with bearing measurements is uniformly observable. This does not hold automatically, as in Proposition~\ref{proposition::tfg_stability_non_bias}. Similar observability analyses involving bearings have been comprehensively conducted in \cite{hybrid_observer_vision_aided,riccati_observer_pnp}. Such analyses are beyond the scope of the present paper and thus are omitted.

Another interesting type of measurement is range, given by
\begin{equation}
    \label{eq::range_measurement}
    y^{(i)}=\left\Vert T_t^{-1}d^{(i)}\right\Vert_2
\end{equation}
for \eqref{eq::tfg_case1_systems}. Hence, the original measurement equation $\bar{y}^{(i)}=\bar{z}_0^{(i)}$ in \eqref{eq::immersion_tfg_right_bar_final} is replaced with
\begin{equation}
    \label{eq::range_immersion}
    \bar{y}^{(i)}=\left\Vert\bar{z_0}^{(i)}\right\Vert_2.
\end{equation}

\begin{proposition}\label{proposition::range_immersion}
    The embedded system \eqref{eq::immersion_tfg_right_bar_final} can be further embedded into an LTV system with state $(\bar{z},s)$ given by
    \begin{equation}
        \label{eq::ltv_range_immersion}
        \dot{\bar{z}}=f_1(\bar{z},u),\ \dot{s}=f_2(\bar{z},s,u),\ \frac{1}{2}\left(y^{(i)}\right)^2=s_{0,0}^{(i)}
    \end{equation} 
    where $f_1$ is linear in $\bar{z}$, as in \eqref{eq::immersion_tfg_right_bar_final}. $f_2$ is linear in $\bar{z}$ and $s$, and $u$ is the input. The extended state $s:=[...,s_{j,k}^{(i)},...]^\top$, where each component is defined by
    \begin{equation}
        s_{j,k}^{(i)}:=\frac{1}{2}\bar{z}_j^{(i)\top}\bar{z}_k^{(i)},
    \end{equation}
    with $0\le j\le k$, $0\le k\le N-1$ and $1\le i\le M$. Recall that $N=d+n+m$.
\end{proposition}

\begin{proof}
    See Appendix~\ref{app::proof_range_immersion} for the expressions of $f_2$.
\end{proof}

\begin{remark}
    A linear time-invariant system with quadratic outputs can be embedded into a higher-dimensional LTV system \cite{LTI_quadratic_outputs, LTI_range_navigation}. In general, if the system dynamics become time-varying, the embedding fails. Proposition~\ref{proposition::range_immersion} holds due to the skew-symmetric structure of $\omega_t^\times$. Hence, we provide a valuable example of a successful embedding of an LTV system with quadratic outputs.
\end{remark}

For non-biased embeddable two-frame systems with range measurements, our observer is implemented with the Kalman filter for the embedded system designed for \eqref{eq::ltv_range_immersion}. In addition to the rank condition on $D$, global exponential stability is achieved if the Kalman observer for \eqref{eq::ltv_range_immersion} is uniformly observable. This does not hold without further assumptions, as in Proposition~\ref{proposition::tfg_stability_non_bias}. Detailed observability analysis for \eqref{eq::ltv_range_immersion} is beyond the scope of the present paper.

\subsection{An Example of an Embeddable Two-Frame System}

Navigation on rotating Earth provides a nontrivial example demonstrating how to fit a system into an embeddable two-frame model. Let $R_t,p_t$ and $v_t$ denote the attitude, position and velocity, respectively. Let the Earth's angular velocity be $\Omega\in\mathbb{R}^{3}$. The unbiased IMU dynamics are given by $\dot{R}_t=-\Omega^\times R_t+R_t\omega_t^\times$, $\dot{p}_t=v_t$ and $\dot{v}_t=R_ta_t+g-2\Omega^\times v_t-(\Omega^\times)^2p_t$, where $\omega_t,a_t\in\mathbb{R}^3$ are the gyroscope and accelerometer inputs \cite{rotating_earth}. To fit this system into the two-frame framework, define $W_t:=[p_t, v_t+\Omega^\times p_t]$ and define the state $T_t=\begin{bmatrix}R_t & W_t\\ 0 & I\end{bmatrix}\in\TFG(3,2,0)$. The IMU dynamics satisfy \eqref{eq::tfg_case1_systems} as $\dot{T}_t=\begin{bmatrix}
    -\Omega^\times & \gamma\\
    0 & L
\end{bmatrix}T_t+T_t\begin{bmatrix}
    \omega_t^\times & \rho_t\\
    0 & -L
\end{bmatrix}$, where $\gamma:=[0, g]\in\mathbb{R}^{3\times 2}$, $\rho_t:=[0, a_t]\in\mathbb{R}^{3\times 2}$. Note that $L:=\begin{bmatrix}
    0 & 0 \\ -1 & 0
\end{bmatrix}$ and $g$ is the local gravity vector expressed in the world frame. Consider two landmark-type measurements $\bar{y}^{(i)}=R_t^{-1}(\bar{d}^{(i)}-p_t),i=1,2$, one bearing measurement $\bar{y}^{(3)}=\pi_{\mathbb{S}^2}(R_t^{-1}(\bar{d}^{(3)}-p_t))$, and one range measurement $\bar{y}^{(4)}=\Vert R_t^{-1}(\bar{d}^{(4)}-p_t)\Vert$. The $\bar{d}^{(i)},i=1,\dots,4$, are known vectors in $\mathbb{R}^3$. To express the measurement equations in the form $y^{(i)}=T_t^{-1}d^{(i)}$, we introduce homogeneous coordinates $y^{(i)}=\begin{bmatrix}
    \bar{y}^{(i)} \\ \underline{y}^{(i)}
\end{bmatrix}, d^{(i)}=\begin{bmatrix}
    \bar{d}^{(i)} \\ \underline{d}^{(i)}
\end{bmatrix}$ where $\underline{d}^{(i)}=\underline{y}^{(i)}=\begin{bmatrix}
    1 \\ 0
\end{bmatrix}$.

This example is a Type-1 embeddable system. Define $d_j^{(i)}$ and the embedded state $z_j^{(i)}$ as $z_j^{(i)}:=T_t^{-1}\begin{bmatrix}
    -\Omega^\times & \gamma\\
    0 & L
\end{bmatrix}^j d^{(i)}:=T_t^{-1}d_j^{(i)}$ with $i\in[1,4]$ and $j\in[0,4]$. Let $z_j^{(i)}=\begin{bmatrix}
    \bar{z}_j^{(i)} \\ \underline{z}_j^{(i)}
\end{bmatrix}$ and $d_j^{(i)}=\begin{bmatrix}
    \bar{d}_j^{(i)} \\ \underline{d}_j^{(i)}
\end{bmatrix}$. The underlined variables are $\underline{d}_0^{(i)}\equiv\underline{z}_0^{(i)}=\begin{bmatrix}
    1 \\ 0
\end{bmatrix}$, $\underline{d}_1^{(i)}\equiv\underline{z}_1^{(i)}=\begin{bmatrix}
    0 \\ -1
\end{bmatrix}$ and $\underline{d}_j^{(i)}\equiv\underline{z}_j^{(i)}=\begin{bmatrix}
    0 \\ 0
\end{bmatrix}(j=2,3,4)$. The barred variables are $\bar{d}^{(i)}_0=\bar{d}^{(i)}$, $\bar{d}^{(i)}_1=-\Omega^\times\bar{d}^{(i)}_0$, $\bar{d}^{(i)}_2=(\Omega^\times)^2\bar{d}^{(i)}_0-g$, $\bar{d}^{(i)}_3=\Vert\Omega\Vert^2\Omega^\times(\bar{d}^{(i)}_0+g)$, and $\bar{d}^{(i)}_4=-\Vert\Omega\Vert^2(\Omega^\times)^2(\bar{d}^{(i)}_0+g)$. By Theorem~\ref{theorem::embedding_ltv}, the embedded LTV system is given by $\dot{\bar{z}}^{(i)}_0=-\omega^\times\bar{z}^{(i)}_0-\bar{z}^{(i)}_1, \dot{\bar{z}}^{(i)}_1=-\omega^\times\bar{z}^{(i)}_1+a-\bar{z}^{(i)}_2, \dot{\bar{z}}^{(i)}_2=-\omega^\times\bar{z}^{(i)}_2-\bar{z}^{(i)}_3, \dot{\bar{z}}^{(i)}_3=-\omega^\times\bar{z}^{(i)}_3-\bar{z}^{(i)}_4, \dot{\bar{z}}^{(i)}_4=-\omega^\times\bar{z}^{(i)}_4+\Vert\Omega\Vert^2\bar{z}^{(i)}_3$ for $i\in[1,4]$. The landmark measurements are $\bar{y}^{(i)}=\bar{z}_0^{(i)},i=1,2$. The bearing measurement is $0=(I_{3\times 3}-\bar{y}^{(3)}\bar{y}^{(3)\top})\bar{z}_0^{(3)}$. For the range measurement, the embedded LTV is extended to include $\dot{s}_{0,0}=-2s_{0,1}, \dot{s}_{0,1}=-s_{1,1}-\frac{1}{2}a^\top\bar{z}^{(4)}_0-s_{0,2}, \dot{s}_{0,2}=-s_{1,2}-s_{0,3}, \dot{s}_{0,3}=-s_{1,3}-s_{0,4}, \dot{s}_{0,4}=-s_{1,4}+\Vert\Omega\Vert^2s_{0,3}, \dot{s}_{1,1}=-a^\top\bar{z}^{(4)}_1-2s_{1,2}, \dot{s}_{1,2}=-\frac{1}{2}a^\top\bar{z}^{(4)}_2-s_{2,2}-s_{1,3}, \dot{s}_{1,3}=-\frac{1}{2}a^\top\bar{z}^{(4)}_3-s_{2,3}-s_{1,4}, \dot{s}_{1,4}=-\frac{1}{2}a^\top\bar{z}^{(4)}_4-s_{2,4}+\Vert\Omega\Vert^2s_{1,3}, \dot{s}_{2,2}=-2s_{2,3}, \dot{s}_{2,3}=-s_{3,3}-s_{2,4}, \dot{s}_{2,4}=-s_{3,4}+\Vert\Omega\Vert^2s_{2,3}, \dot{s}_{3,3}=-2s_{3,4}, \dot{s}_{3,4}=-s_{4,4}+\Vert\Omega\Vert^2s_{3,3}, \dot{s}_{4,4}=2\Vert\Omega\Vert^2s_{3,4}$, where each $s_{\cdot,\cdot}$ is a scalar. The range measurement is $\frac{1}{2}(y^{(4)})^2=s_{0,0}$, which is linear. We use a Kalman observer to obtain estimates of the embedded state. Using these estimates, we employ \eqref{eq::tfg_state_reconstruct_right} to reconstruct the $\TFG(3,2,0)$ state. The matrices $\hatbar{Z},\bar{D}$, and $\underline{D}$ are assembled using previously defined constants following \eqref{eq::bar_Z}--\eqref{eq::underline_D}.

\section{Simulation on Examples}
\label{sec::simulation_examples}

We present two pedagogical examples based on the classical systems in Section~\ref{subsec::att_observer} and Section~\ref{subsec::imu_lmk_observer}. We emphasize that embeddable two-frame systems encompass much more than these examples.

\subsection{Embedding-based Attitude Observer}\label{subsec::sim_att_obs}

Returning to the model in Section~\ref{subsec::att_observer}, we consider only two landmarks and additionally estimate gyroscope bias. Including noise, the system equations are
\begin{equation}
    \label{eq::att_obs_with_noise}
    \left\{\begin{aligned}
    \dot{z}_0^{(i)}&=-(\omega-b_g+n_g)^\times z_0^{(i)}\\
    \dot{b}_g&=n_{b_g},\ \ y^{(i)}=z_0^{(i)}+n_y
    \end{aligned}\right.,
\end{equation}
where $b_g\in\mathbb{R}^3$ is the bias, and $i=1,2$. $n_g$, $n_{b_g}$, and $n_y$ are the gyroscope, bias, and measurement noise, respectively. Define the error-state components as $z_0^{(i)}=\hat{z}_0^{(i)}+\delta z_0^{(i)}$, $b_g=\hat{b}_g+\delta b_g$. Let the error-state vector be $[\delta z_0^{(1)\top},\delta z_0^{(2)\top},\delta b_g^\top]^\top$ and let the process noise be $[n_g^\top, n_{b_g}^\top]^\top$. The Kalman observer for \eqref{eq::att_obs_with_noise} is designed following Section~\ref{subsec::ekf} as
\begin{align}
    &F=\begin{bmatrix}
        -(\omega-\hat{b}_g)^\times & 0_3 & -\left(\hat{z}_0^{(1)}\right)^\times\\
        0_3 & -(\omega-\hat{b}_g)^\times & -\left(\hat{z}_0^{(2)}\right)^\times\\
        0_3 & 0_3 & 0_3
    \end{bmatrix},\\
    &\mathcal{G}=\begin{bmatrix}
        \left(\hat{z}_0^{(1)}\right) & 0_3\\
        \left(\hat{z}_0^{(2)}\right)^\times & 0_{3}\\
        0_3 & I_3
    \end{bmatrix},\ 
    H=\begin{bmatrix}
        I_3 & 0_3 & 0_3\\
        0_3 & I_3 & 0_3
    \end{bmatrix}.
\end{align}

Using the estimates of this Kalman observer, we obtain estimates of $b_g$, and the attitude is reconstructed using $\hat{Z}=\left[\hat{z}_0^{(1)},\hat{z}_0^{(2)},\left(\hat{z}_0^{(1)}\right)^\times \hat{z}_0^{(2)}\right]$ and $D=[d_1,d_2,d_1^\times d_2]$. In practice, the Mahalanobis norm is used in the optimization weighted by $\Sigma:=\text{diag}(\sigma_1,\sigma_2,\sigma_3)$, where $\sigma_i=\tr(P_i)$ and $P_i$ is the $3\times 3$-diagonal block of $P$ corresponding to each error state. 

Simulations are conducted to compare the proposed observer with the imperfect InEKF with bias. We use a right-invariant error for the attitude and an additive error for the bias, given by $R=\exp(\delta\theta^\times)\hat{R}$ and $b_g=\hat{b}_g+\delta b_g$. The InEKF is tuned using the same noise parameters $n_g$, $n_{b_g}$, and $n_y$, with corresponding Jacobians
\begin{equation}
    F=\begin{bmatrix}
        0_3 & -\hat{R}\\
        0_3 & 0_3
    \end{bmatrix},\ \mathcal{G}=\begin{bmatrix}
        \hat{R} & 0_3\\
        0_3 & I_3
    \end{bmatrix},\ 
    H=\begin{bmatrix}
        d_1^\times & 0_3\\
        d_2^\times & 0_3
    \end{bmatrix},
\end{equation}
where the error-state and process noise are $[\delta\theta^\top,\delta b_g^\top]^\top$ and $[n_g^\top,n_{b_g}^\top]^\top$. We generate the system trajectory using closed-form formulas. Let $\alpha(t)=\pi\sin(\frac{\pi}{40}t)$, $\beta(t)= 2\pi\cos(\frac{\pi}{30}t+\frac{\pi}{9})$, and $\gamma(t)=2\pi\sin(\frac{\pi}{25}t-\frac{\pi}{7})$. The rotation axis is $n(t)=[\cos\beta(t)\cos\gamma(t), \cos\beta(t)\sin\gamma(t), \sin\beta(t)]^\top$, and the ground-truth rotation is $R(t)=\exp(\alpha(t)n(t)^\times)$. Let $b_g=[0.02, -0.01, 0.01]^\top$. The gyroscope input $\omega$ is calculated from the true trajectory. The positions of the two landmarks are $d_1=[-5,10,3]^\top$ and $d_2=[6,0,-5]^\top$. For the noise, we set $Q=\text{diag}(\text{cov}(n_g)I_3,\text{cov}(n_{b_g}))=\text{diag}(10^{-2}I_3,10^{-4}I_3)$ and $\mathcal{R}=I_2\otimes\text{cov}(n_y)=1.0I_2\otimes I_3$. The IMU rate is $200\text{Hz}$, and updates are performed every three IMU samples. Both the proposed observer and the InEKF share identical inputs, outputs, and noise configurations. The two methods are initialized with the same large initial error, specifically, a rotation angle error of $0.99\pi$ about the axis $[0.59,0.43,0.68]^\top$, and zero initial bias. In this setting, $\Vert\omega\Vert$ and $\Vert\dot{\omega}\Vert$ are not identically zero, and the embedded system is numerically verified to be uniformly observable. The results are shown in Fig.~\ref{fig::sim_att_obs}, demonstrating the superior convergence performance of the proposed method.

\begin{figure}[t]
 \centering
 \parbox{3.3in}{
   \includegraphics[scale=0.65]{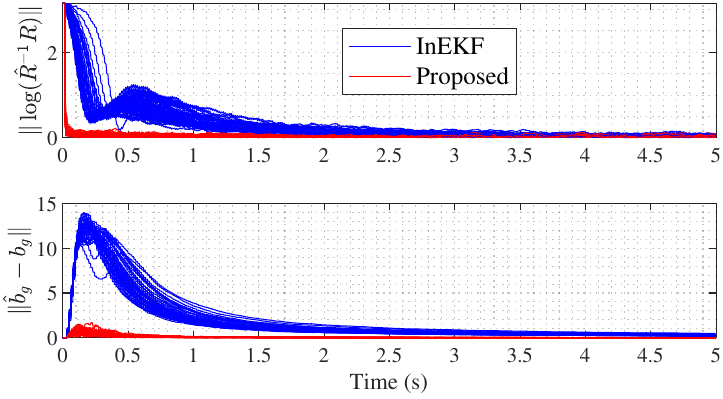}
 }
 \caption{Comparison of the proposed observer with the InEKF for attitude estimation. 50 runs are conducted from the same initial condition with large initial rotation error and identical noise configurations.} 
 \label{fig::sim_att_obs}
 \vspace{-0.5cm}
\end{figure}

\subsection{Embedding-based IMU-Landmark Pose Observer}

Continuing with the system in Section~\ref{subsec::imu_lmk_observer}, we consider three landmarks. Associating noise to \eqref{eq::imu_lmk_ltv}, we obtain the system equations
\begin{equation}
    \label{eq::imu_lmk_noise_ltv}
    \left\{\begin{aligned}
        \dot{z}_0^{(i)}&=-(\omega_t+n_g)^\times z_0^{(i)}-z_1\\
        \dot{z}_1&=-(\omega_t+n_g)^\times z_1-z_2+a_t+n_a\\
        \dot{z}_2&=-(\omega_t+n_g)^\times z_2,\ \ y_t^{(i)}=z_0^{(i)}+n_y
    \end{aligned}\right.,
\end{equation}
where $i=1,2,3$. $n_g$, $n_a$, and $n_y$ are $\mathbb{R}^3$-valued gyroscope, accelerometer, and landmark measurement noise, respectively. Define the error-state components $z_0^{(i)}=\hat{z}_0^{(i)}+\delta z_0^{(i)}$, $z_1=\hat{z}_1+\delta z_1$, and $z_2=\hat{z}_2+\delta z_2$. Let the error-state vector be $[z_0^{(1)\top},z_0^{(2)\top},z_0^{(3)\top},z_1^\top,z_2^\top]^\top$ and let the process noise be $[n_g^\top, n_a^\top]^\top$. Stacking the measurements, the Kalman observer for \eqref{eq::imu_lmk_noise_ltv}, designed following Section~\ref{subsec::ekf}, is given by
\begin{align}
    &F=-\omega_t^\times\otimes I_5+I_3\otimes J_5^{(1)},\ H=[I_9,\ 0_{9\times 15}],\\
    &\mathcal{G}=\begin{bmatrix}
        \left(\hat{z}_0^{(1)}\right)^\times & \left(\hat{z}_0^{(2)}\right)^\times & \left(\hat{z}_0^{(3)}\right)^\times & \hat{z}_1^\times & \hat{z}_2^\times\\
        0_3 & 0_3 & 0_3 & I_3 & 0_3 
    \end{bmatrix}^\top.
\end{align}
where $J_5^{(1)}\in\mathbb{R}^{5\times 5}$ has ones on the first upper off-diagonal and zeros everywhere else. State reconstruction is performed using \eqref{eq::imu_lmk_opt}--\eqref{eq::imu_lmk_pos_recon}. We conduct simulations to compare the proposed method to the InEKF. We use a right-invariant error on $\text{SE}_2(3)$ for the pose, given by $R=\exp(\delta\theta^\times)\hat{R}$, $p=\exp(\delta\theta)\hat{p}+\mathcal{J}_l(\delta\theta)\delta p$, and $v=\exp(\delta\theta)\hat{v}+\mathcal{J}_l(\delta\theta)\delta v$, where $\mathcal{J}_l$ is the left Jacobian of $\SO(3)$. The InEKF is tuned using the same noise parameters $n_g$ and $n_a$, with corresponding Jacobians
\begin{equation}
    F=\begin{bmatrix}
        0_3 & 0_3 & 0_3\\
        0_3 & 0_3 & I_3\\
        g^\times & 0_3 & 0_3
    \end{bmatrix},\mathcal{G}=\begin{bmatrix}
        \hat{R} & 0_3\\
        \hat{p}^\times\hat{R} & 0_3\\
        \hat{v}^\times\hat{R} & 0_3
    \end{bmatrix},H=\begin{bmatrix}
        H_1\\
        H_2\\
        H_3
    \end{bmatrix} ,
\end{equation}
where $H_i=\left[d_i^\times,\ -I_3,\ 0_3\right],\ i=1,2,3$. The error-state and process noise for the above Jacobians are $[\delta\theta^\top,\delta p^\top,\delta v^\top]^\top$ and $[n_g^\top,n_a^\top]^\top$. Closed-form formulas are used to generate the system trajectory. We use $\alpha(t)$, $\beta(t)$, and $\gamma(t)$ from Section~\ref{subsec::sim_att_obs} to generate the same ground-truth attitude. Let $p_x(t)=20\cos(\frac{\pi}{55}t)-5$, $p_y(t)= 40\sin(\frac{\pi}{65}t)$ and $p_z(t)=60\sin(\frac{\pi}{50}t)$. The ground-truth velocity and acceleration are calculated. The gravity vector is $g=[0,0,-9.81]^\top$. The positions of the three landmarks are $d_1=[-20,1,19]^\top$, $d_2=[-33,-30,5]^\top$, and $d_3=[24,60,-70]^\top$. For the noise, we set $Q=\text{diag}(\text{cov}(n_g),\text{cov}(n_a))=\text{diag}(0.1I_3,0.32I_3)$ and $\mathcal{R}=I_3\otimes\text{cov}(n_y)=1.0 I_3\otimes I_3$. As before, the IMU rate is $200\text{Hz}$ and updates are performed every 3 steps. The proposed observer and the InEKF share identical inputs, outputs and noise configurations. Their initializations are also identical. The initial rotation error is the same as in Section~\ref{subsec::sim_att_obs}. The initial position and velocity errors in the XYZ-directions are $[25,25,25]^\top$ and $[-15,15,15]^\top$, respectively. Fig.~\ref{fig::sim_imu_lmk_obs} demonstrates the superior convergence of the proposed method. If the accelerometer inputs are corrupted by a constant bias, GES is preserved when the bias is simultaneously estimated. If, in addition, there is a constant gyroscope bias, we can only achieve semi-global stability with sufficient angular excitation to ensure the system is uniformly observable, in addition to the previous assumptions. 

\begin{figure}[t]
 \centering
 \parbox{3.3in}{
   \includegraphics[scale=0.65]{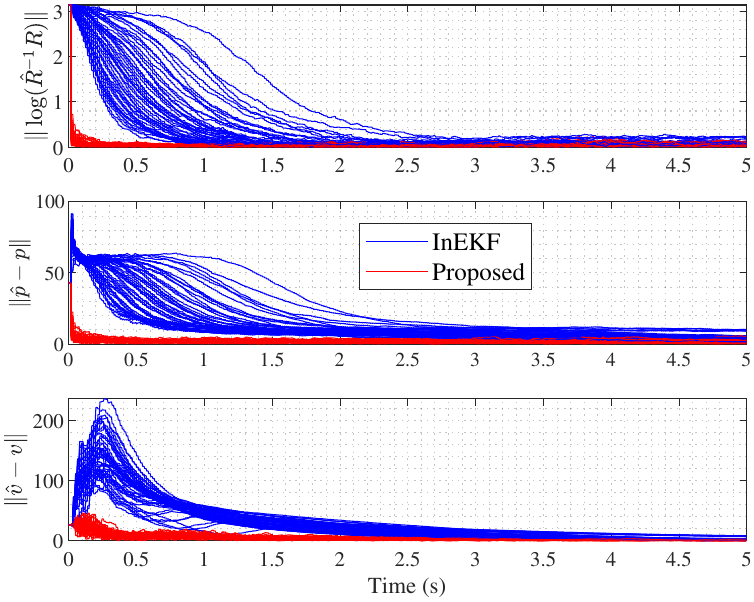}
 }
 \caption{Comparison of the proposed method with the InEKF for IMU-landmark navigation. 50 runs are conducted from the same initial condition.} 
 \label{fig::sim_imu_lmk_obs}
 \vspace{-0.5cm}
\end{figure}

\section{Conclusion}

We identify a class of LTV embeddable linear observed systems. We propose an observer framework consisting of a Kalman-like observer followed by optimization-based state reconstruction. GES is achieved provided that a suitable observability condition holds and a global minimum of the optimization is attained. The assumptions required for GES are quantitatively compared with those of the InEKF. Bias estimation with semi-global stability guarantees is discussed in detail. We apply the theory to two-frame systems. An interesting aspect is that the optimization-based state reconstruction is performed at each time step when the Kalman estimator provides an estimate of the embedded LTV, raising computational cost concerns. Future work includes developing saddle-escaping dynamics to alleviate the computational burden of state reconstruction.

\appendix

\subsection{Proof of Theorem~\ref{theorem::embedding_ltv}}\label{app::proof_embedding_ltv}

For a Type-1 system, the dynamics read $\dot\chi=f_u(\chi)=g_{\tilde{A}}(\chi)+\chi f_u(id)=\tilde{A}\chi-\chi (\tilde{A}-f_u(id))$ with measurement $y=\chi^{-1}d\in\mathbb{R}^{d_y}$. Using the pushforward of the linear group action, define $z_0:=\chi^{-1} d\in\mathbb{R}^{d_y}$ and calculate $\mathcal{L}_{f_u}h$ as
\begin{align*}
    \dot{z}_0&=-\chi^{-1}\dot\chi\chi^{-1}d=-\chi^{-1}(\tilde{A}\chi+\chi (f_u(id)-\tilde{A}))\chi^{-1}d\\
    &=-S_u(\chi^{-1} d)-\chi^{-1}(\tilde{A}d)=-S_uz_0-\chi^{-1}(\tilde{A}d). 
\end{align*}
Note that the measurement equation is $y=z_0$. Define $z_j:=\chi^{-1}(\tilde{A}^jd)$. Taking the time derivative along $f_u$, we obtain
\begin{align*}
    \dot{z}_j&=-\chi^{-1}\dot\chi\chi^{-1}(\tilde{A}^jd)=-\chi^{-1}(\tilde{A}\chi+\chi S_u)\chi^{-1}(\tilde{A}^jd)\\
    &=-S_u\chi^{-1}(\tilde{A}^jd)-\chi^{-1}(\tilde{A}^{j+1}d)=-S_u z_j-z_{j+1},
\end{align*}
with $S_u:=f_u(id)-\tilde{A}$. Note that we have used the fact that $\tilde{A}\in\mathfrak{g}$ does not depend on $u$ or time. It's natural to ask whether the inductive definition of $z_{j+1}$ from $z_j$ terminates after finitely many steps. As $\tilde{A}\in\mathfrak{g}$ is now viewed as a time-independent $\mathbb{R}^{d_y\times d_y}$ linear operator, by the Cayley-Hamilton theorem, there exist real constants $a_0,a_1,\dots,a_{d_y-1}$ such that
\begin{equation}
    \label{eq::linear_combination_A}
    \tilde{A}^{d_y}=a_{d_y-1}\tilde{A}^{d_y-1}+a_{d_y-2}\tilde{A}^{d_y-2}+\dots+a_1\tilde{A}+a_0I.
\end{equation} 
By the definition of the $z_j$, this directly yields
\begin{equation}
    \label{eq::linear_combination_z}
    z_{d_y}=a_{d_y-1}z_{d_y-1}+a_{d_y-2}z_{d_y-2}+\dots+a_1z_1+a_0z_0.
\end{equation}
To summarize, a closed set of equations for the variables $z_0,...,z_{d_y-1}$ is obtained with measurement $y=z_0$:
\begin{equation}
    \left\{\begin{aligned}
    \dot z_j&=-S_u z_j-z_{j+1},\ 0\le j\le d_y-2\\
    \dot z_{d_y-1}&=-S_u z_{d_y-1}-z_{d_y}=-S_u z_{d_y-1}-\sum_{l=0}^{d_y-1}a_lz_l.
    \end{aligned}\right..
\end{equation}
The above equation is linear time-varying, completing the embedding for Type-1 systems. Although the mechanics for Type-2 systems are analogous, we emphasize that to cancel the undesirable terms involving $\chi$ by multiplication as before, the group-affine dynamics should be adjusted to $g_{\tilde{A}}(\chi)+f_u(id)\chi$ and the measurement takes the form of a left linear action. Let $z_j:=\chi(\tilde{A}^jd)$. Taking derivatives along the new $f_u$ yields
\begin{align*}
    \dot{z}_j&=\dot\chi(\tilde{A}^jd)=-S_u\chi(\tilde{A}^jd)-\chi(\tilde{A}^{j+1}d)=-S_u z_j-z_{j+1}.
\end{align*}
with $S_u=-f_u(id)-\tilde{A}$. As the operator $\tilde{A}\in\mathfrak{g}$ satisfies the same constraint \eqref{eq::linear_combination_A}, this leads to the same relationship \eqref{eq::linear_combination_z} among $z_j$s for Type-2 systems. The embedded LTV system for Type-2 systems is the same as that for Type-1 systems, with the newly defined $S_u$ and the same coefficients $a_l$. Note the output equation is still $y=z_0$.
\endproof

\subsection{Proof of Lemma~\ref{lm::FH_observability}}\label{app::proof_FH_observability}

Consider the embedded LTV for the $i$-th output given by $\dot{z}^{(i)}=F_u^{(i)}z^{(i)}+C_u^{(i)},y^{(i)}=H^{(i)}z^{(i)}$. It suffices to prove that the pair $(F_u^{(i)},H^{(i)})$ is uniformly observable. The structures of $F_u^{(i)}$ and $H^{(i)}$ follow from \eqref{eq::embedded_LTV}. 

We first calculate the transition matrix $\Phi^F(t_2,t_1)$ of $F_u^{(i)}$. Let $F_u^{(i)}=\bar{A}+\bar{S}_u$, where $\bar{S}_u=\text{diag}(-S_u,\dots,-S_u)$ and
\begin{align*}
    \bar{A}=\begin{bmatrix}
        0 & -I & 0 & \cdots & 0\\
        0 & 0 & -I & \cdots & 0\\
        \vdots & \vdots & \vdots & \vdots & \vdots \\
        -\tilde{a}_0I & -\tilde{a}_1I & -\tilde{a}_2I & \cdots & -\tilde{a}_{d_y-1}I
    \end{bmatrix}\in\mathbb{R}^{d_y^2\times d_y^2}.
\end{align*}
Note that the subscript $u$ in $\bar{S}_u$ indicates its dependence on the input. It can be verified that $\bar{S}_u\bar{A}=\bar{A}\bar{S}_u$. The decomposition of $F_u^{(i)}$ into the sum of two commuting matrices is the standard technique in observability analysis \cite{hybrid_observer_landmark,hybrid_observer_vision_aided}. Let $\Psi(t)=\text{diag}(\psi(t),\dots,\psi(t))$. We claim that $\Phi^F(t_2,t_1)=\Psi(t_2)\Phi^{\bar{A}}(t_2,t_1)\Psi(t_1)^{-1}$, where $\Phi^{\bar{A}}(\cdot,\cdot)$ is the transition matrix of $\bar{A}$. It is clear that $\Phi^{F}(t_1,t_1)=I$ and $\Phi^F(t_2,t_1)=\Phi^F(t_1,t_2)^{-1}$. Computing the partial derivative $\frac{\partial\Phi^F}{\partial t_2}=\bar{S}_u\Phi^F+\Psi(t_2)\bar{A}\Phi^{\bar{A}}(t_2,t_1)\Psi^{-1}(t_1)=(\bar{S}_u+\bar{A})\Phi^F(t_2)$, thereby proving the claim. The constant matrix pair $(\bar{A},\bar{H})$ is observable, as $[\bar{H}^\top,(\bar{H}\bar{A})^\top,...,(\bar{H}\bar{A}^N)^\top]^\top$ is of full column rank, implying there exist $\delta,\alpha_3>0$ such that $\int_{t}^{t+\delta}\Phi^{\bar{A}}(\tau,t)^\top\bar{H}_\tau^\top R^{-1}\bar{H}_\tau\Phi^{\bar{A}}(\tau,t)d\tau\succeq\alpha_3I,\forall t$, hence $\int_t^{t+\delta}\Phi^{\bar{A}}(\tau,t)^\top\Phi^{\bar{A}}(\tau,t)d\tau\succeq \lambda_{\min}(R)I$. As $\mathcal{O}(t+\delta,t)=$
\begin{align*}
&\int_t^{t+\delta}\Phi^F(\tau,t)^\top H^{(i)\top}R_\tau^{-1}H^{(i)}\Phi^F(\tau,t) d\tau\\
&\succeq \frac{\alpha_1}{\alpha_2\lambda_{\max}(R)}\int_t^{t+\delta}\Phi^{\bar{A}}(\tau,t)^\top\Phi^{\bar{A}}(\tau,t)d\tau\succeq\frac{\alpha_1\alpha_3}{\alpha_2\lambda_{\max}(R)}I
\end{align*}
for each $t\in\mathbb{R}$, completing the proof.
\endproof

\subsection{Proof of Theorem~\ref{theorem::stability_non_bias}}\label{app::proof_stability_non_bias}
    We prove the result only for Type-1 systems. Let $\chi(t)$ and $\hat{\chi}(t)$ be the true and estimated trajectories on $G$, respectively. Let $z=\pi(\chi)$ and $\hat{z}=\pi(\hat{\chi})$ be the true and estimated states of the embedded LTV, where $\pi$ is the embedding map. Let $N=d_y-1$. Let $Z=[z^{(1)}_0,...,z^{(1)}_N,...,z^{(i)}_j,...,z^{(M)}_0,...,z^{(M)}_N]$ and $\hat{Z}=[\hat{z}^{(1)}_0,...,\hat{z}^{(1)}_N,...,\hat{z}^{(i)}_j,...,\hat{z}^{(M)}_0,...,\hat{z}^{(M)}_N]$ associated with the true and estimated states of the LTV. Let the system structure matrix be $D=[d^{(1)}_0,...,d^{(1)}_N,...,d^{(i)}_j,...,d^{(M)}_0,...,d^{(M)}_N]$. The matrices $Z$, $\hat{Z}$, and $D$ are in $\mathbb{R}^{d_y\times Md_y}$. 
    
    We first show that $d(\hat{\chi},\chi)$ can be bounded by $\Vert\hat{Z}-Z\Vert$. A well-known fact is that $\lambda_{\min}(S)\tr(K)\le\tr(KS)\le\lambda_{\max}(S)\tr(K)$, where $K\in\mathbb{S}_+$ and $S$ is a symmetric matrix of the same size as $K$. Hence, $\lambda_{\min}(DD^\top)\tr[(\chi^{-1}-\hat{\chi}^{-1})(\chi^{-1}-\hat{\chi}^{-1})^\top]\le\tr[DD^\top(\chi^{-1}-\hat{\chi}^{-1})(\chi^{-1}-\hat{\chi}^{-1})^\top]=\tr[(\chi^{-1}-\hat{\chi}^{-1})DD^\top(\chi^{-1}-\hat{\chi}^{-1})^\top]=\Vert\chi^{-1}D-\hat{\chi}^{-1}D\Vert^2=\Vert Z-\hat{\chi}^{-1}D\Vert^2$. Moreover, since $\hat{\chi}=\mathop{\arg\min}_{\chi'\in G}\Vert\hat{Z}-\chi'^{-1}D\Vert^2$, we have $\Vert\hat{Z}-\hat{\chi}^{-1}D\Vert\le\Vert\hat{Z}-Z\Vert$, since there exists $\chi\in G$ such that $Z=\chi^{-1}D$. This yields $\Vert Z-\hat{\chi}^{-1}D\Vert\le\Vert Z-\hat{Z}\Vert+\Vert\hat{Z}-\hat{\chi}^{-1}D\Vert\le 2\Vert\hat{Z}-Z\Vert$. By Assumption~\ref{assumption::structure_rank_condition}, we have $\sigma_{\min}(D)>0$, which implies $d(\hat{\chi},\chi)\le\frac{2}{\sigma_{\min}(D)}\Vert\hat{Z}-Z\Vert$.
    
    We now show that $\hat{Z}$ converges to $Z$ exponentially from any initial condition. This follows the standard proof of the stability of the Kalman observer under uniform observability \cite{riccati_observer_pnp,position_range_riccati}, which we briefly sketch. By Assumption~\ref{assumption::lower_boundedness_transition}, $P$ is uniformly bounded, i.e., $p_mI\preceq P\preceq p_MI$ \cite{boundedness_riccati_ode}. Let $\tilde{z}=\hat{z}-z$ denote the error state, which satisfies $\dot{\tilde z}=(F_u-KH)\tilde z$. Define $V(\tilde z):=\tilde{z}^\top P^{-1}\tilde{z}$. We obtain $\frac{1}{p_M}\Vert\tilde{z}\Vert^2\le V\le \frac{1}{p_m}\Vert\tilde{z}\Vert^2$ and $V\rightarrow+\infty$ as $\Vert\tilde z\Vert\rightarrow+\infty$. To prove uniform GES, it suffices to show that $\dot{V}\le-\frac{\lambda_{\min}(Q)}{p_M^2}\Vert\tilde{z}\Vert^2$. There exist constants $c_1,c_2>0$ such that $\Vert\hat{Z}(t)-Z(t)\Vert\le c_1\Vert\hat{Z}(t_0)-Z(t_0)\Vert\mathrm{e}^{-c_2(t-t_0)}$. Since $c_3d(\hat{\chi}(t_0),\chi(t_0))\ge\Vert\hat{Z}(t_0)-Z(t_0)\Vert$ for some $c_3>0$, the error metric satisfies $d(\hat{\chi}(t),\chi(t))\le\frac{2c_1c_3}{\sigma_{\min}(D)}d(\hat{\chi}(t_0),\chi(t_0))\mathrm{e}^{-c_2(t-t_0)}$, completing the proof.
\endproof

\subsection{Proof of Theorem~\ref{th::obs_comparison}}\label{app::proof_obs_comparison}
    We claim that $\mathcal{O}_2x=0$ for some $x\in\mathbb{R}^{\dim\mathfrak{g}}$ if and only if $\mathcal{L}(x)\tilde{A}^kd^{(i)}=0$ for $0\le k\le d_y-1$ and $1\le i\le M$. First, we prove the necessity. $\mathcal{O}_2x=0$ implies that $\mathcal{L}^\dagger(d^{(i)})ad_{\mathcal{L}^{-1}(\tilde{A})}^jx=0$ for $0\le j\le\dim\mathfrak{g}-1$ and $1\le i\le M$. For $k=0$, we have $\mathcal{L}(x)\tilde{A}^0d^{(i)}=-\mathcal{L}^\dagger(d^{(i)})x=0$. We proceed by induction on $k$. Suppose that $\mathcal{L}(x)\tilde{A}^jd^{(i)}=0$ for $0\le j\le k-1$. The $k$-th order block of $\mathcal{O}_2$ yields $0=\mathcal{L}^\dagger(d^{(i)})ad_{\mathcal{L}^{-1}(\tilde{A})}^kx=-\mathcal{L}(ad^k_{\mathcal{L}^{-1}(\tilde{A})}x)d^{(i)}=-\mathcal{L}([\mathcal{L}^{-1}(\tilde{A}),ad_{\mathcal{L}^{-1}(\tilde{A})}^{k-1}x])d^{(i)}=-[\tilde{A},\mathcal{L}(ad_{\mathcal{L}^{-1}(\tilde{A})}^{k-1}x)]d^{(i)}=-\tilde{A}\mathcal{L}(ad_{\mathcal{L}^{-1}(\tilde{A})}^{k-1}x)d^{(i)}+\mathcal{L}(ad_{\mathcal{L}^{-1}(\tilde{A})}^{k-1}x)\tilde{A}d^{(i)}=\cdots=\sum^{j_1\ge 1,j_2\ge 0}_{j_1+j_2=k} a_{j_1,j_2}\tilde{A}^{j_1}\mathcal{L}(x)\tilde{A}^{j_2}d^{(i)}+(-1)^{k}\mathcal{L}(x)\tilde{A}^kd^{(i)}$. By the induction hypothesis, we have $\mathcal{L}(x)\tilde{A}^{k}d^{(i)}=0$, hence the claim holds for all $0\le k\le d_y-1$. Sufficiency follows from the above equation, which expresses $\mathcal{L}(x)\tilde{A}^{k}d^{(i)}$ as a sum of lower-order terms.

    Now assume that $\text{rank}(\mathcal{O}_1)$ is full. Suppose that $\text{rank}(\mathcal{O}_2)$ is not full. Then there exists a nonzero $x\in\mathbb{R}^{\dim\mathfrak{g}}$ satisfying $\mathcal{O}_2x=0$, which implies $\mathcal{L}(x)\tilde{A}^kd^{(i)}=0$ by our claim. Hence, $\mathcal{L}(x)\mathcal{O}_1=0$. Since the column rank of $\mathcal{O}_1$ is full, this implies that the linear operator $\mathcal{L}(x)=0$. Since $\mathcal{L}$ is an isomorphism, $x=0$, contradicting the assumption that $x\neq0$. We conclude that $\text{rank}(\mathcal{O}_2)$ is full. 
    
    Now assume that $\text{rank}(\mathcal{O}_2)$ is full and the addition condition holds. Suppose that the column rank of $\mathcal{O}_1$ is not full. Let $S=\text{span}\mathcal{O}_1\subsetneq\mathbb{R}^{d_y}$, which is $\tilde{A}$-invariant by the Cayley-Hamilton theorem. Since $MS=0$, we have $M\tilde{A}^kd^{(i)}=0$ for $0\le k\le d_y-1$ and $1\le i\le M$. Thus $\mathcal{O}_2\mathcal{L}^{-1}(M)=0$, contradicting the fact that $\text{rank}(\mathcal{O}_2)$ is full. Hence, $\text{rank}(\mathcal{O}_1)$ is full.
\endproof

\subsection{Proof of Lemma~\ref{lm::uniform_observability}}\label{app:proof_uniform_observability}
    In this proof, we use the spectral norm for matrices. Since we consider only one landmark, we omit the index $i$. We first introduce a uniformly bounded invertible state transformation $\eta_j=\psi^{-1}z_j$. Define $\Lambda(b,t):=\psi^{-1}\mathcal{L}(b)\psi$. Using \eqref{eq::embedded_system_bias} for Type-1 embeddable systems, we obtain
    \begin{align}
        \label{eq::transformed_bias_sys}
        \left\{\begin{aligned}
        \dot{\eta}_j&=-\psi^{-1}\dot{\psi}\psi^{-1}z_j+\psi^{-1}\dot{z}_j=\psi^{-1}(\dot{z}_j+S_uz_j)\\
        &=\psi^{-1}(\mathcal{L}(b)z_j-z_{j+1})=\Lambda(b,t)\eta_j-\eta_{j+1}\\
        \dot{\eta}_{N-1}&=\Lambda(b,t)\eta_{N-1}-\sum_{l=0}^{N-1}a_l\eta_l\\
        \dot{b}&=0, \quad y=\psi\eta_0
        \end{aligned}\right.,
    \end{align}
    for $0\le j\le N-2$. With $\eta=[\eta_0^\top...,\eta_j^\top,...,\eta_{N-1}^\top]^\top$, it suffices to establish the uniform observability of \eqref{eq::transformed_bias_sys}. The dynamics are $\dot{\eta}=(\bar{A}+\Delta(b,t))\eta$, where we introduce the notation $\Delta(b,t):=\text{diag}(\Lambda(b,t),...,\Lambda(b,t))$ and the constant matrix
    \begin{equation}
        \bar{A}=\begin{bmatrix}
        0 & -I & 0 & \cdots & 0\\
        0 & 0 & -I & \cdots & 0\\
        \vdots & \vdots & \vdots & \vdots & \vdots \\
        -\tilde{a}_0I & -\tilde{a}_1I & -\tilde{a}_2I & \cdots & -\tilde{a}_{d_y-1}I
    \end{bmatrix}.
    \end{equation}
    Define the extended state $\xi:=[\eta^\top,b^\top]^\top$. We calculate the Jacobians $\delta\dot{\xi}=F_\xi\delta\xi,\ \delta y=H_\xi\delta\xi$ of \eqref{eq::transformed_bias_sys} at $\hat{\xi}$. We obtain
    \begin{equation}
        F_\xi=\begin{bmatrix}
            \bar{A}+\Delta(\hat{b},t) & \bar{\Gamma}(t)\\
            0 & 0
        \end{bmatrix},\ H_\xi=\begin{bmatrix}
            \psi(t) H_\eta & 0
        \end{bmatrix},
    \end{equation}
    where $H_\eta=[I_{d_y},0,...,0]\in\mathbb{R}^{d_y\times d_y^2}$ and
    \begin{equation}
        \bar{\Gamma}(t)=\begin{bmatrix}
            \psi^{-1}\mathcal{L}^{\ddagger}(z_0^{(i)})\\
            \vdots\\
            \psi^{-1}\mathcal{L}^{\ddagger}(z_{N-1}^{(i)})
        \end{bmatrix}=\text{diag}(\psi^{-1},...,\psi^{-1})\Gamma(t).
    \end{equation}
    The observability Gramian $\mathcal{O}(t+\delta,t)$ takes the form
    \begin{equation*}
        \begin{bmatrix}\mathcal{O}_{\eta\eta} & \mathcal{O}_{\eta b}\\
        \mathcal{O}_{\eta b}^\top & \mathcal{O}_{bb}    
        \end{bmatrix}=\int_t^{t+\delta}\Phi_\xi^\top(\tau,t)H_\xi(\tau)R_\tau^{-1}H_\xi(\tau)\Phi_\xi(\tau,t)d\tau.
    \end{equation*}
    Let $F_\eta=\bar{A}+\Delta(\hat{b},t)$, then we have the bound
    \begin{equation}
        \label{eq::bound_F_eta}
        \Vert F_\eta\Vert\le\Vert\bar{A}\Vert+\Vert\psi^{-1}\mathcal{L}(\hat{b})\psi\Vert\le\Vert\bar{A}\Vert+\sqrt{\frac{\alpha_2}{\alpha_1}}\mu_2=c_\eta.
    \end{equation}
    Decomposing the transition matrix $\Phi_\xi$ of $F_\xi$ into blocks, we immediately obtain the bound $e^{-c_\eta(\tau-t)}\le\Vert\Phi_{\eta\eta}(\tau,t)\Vert\le e^{c_\eta(\tau-t)}$ using \eqref{eq::bound_F_eta} and the relationship
    \begin{equation*}
        \Phi_\xi=\begin{bmatrix}
            \Phi_{\eta\eta} & \Phi_{\eta b}\\
            \Phi^\top_{\eta b} & I
        \end{bmatrix},\ \Phi_{\eta b}(\tau,t)=\int_t^{\tau}\Phi_{\eta\eta}(\tau,s)\bar{\Gamma}(s)ds.
    \end{equation*}
    We now estimate the minimum eigenvalue of $\mathcal{O}_{\eta\eta},\mathcal{O}_{bb}$ and the norm of $\mathcal{O}_{\eta b}$ as follows
    \begin{align*}
        \mathcal{O}_{\eta\eta}&=\int_t^{t+\delta}\Phi^\top_{\eta\eta}(\tau,t)H_\eta^\top\psi^\top(\tau)R_\tau^{-1}\psi(\tau)H_\eta\Phi_{\eta\eta}(\tau,t)\\
        &\succeq\frac{\alpha_1}{\lambda_{\max}(R)}\int_t^{t+\delta}\Phi^\top_{\eta\eta}(\tau,t)\Phi_{\eta\eta}(\tau,t)d\tau\succeq\frac{\alpha_1\delta e^{-2\delta c_\eta}}{\lambda_{\max}(R)} I,\\
        \mathcal{O}_{bb}&=\int_t^{t+\delta}\Phi^\top_{\eta b}(\tau,t)H^\top_\eta\psi^\top(\tau)R_\tau^{-1}\psi(\tau)H_\eta\Phi_{\eta b}(\tau,t)d\tau\\
        &\succeq\frac{\alpha_1}{\lambda_{\max}(R)}\int_t^{t+\delta}\int_t^\tau\int_t^\tau\bar{\Gamma}^\top(s_1)\Phi^\top_{\eta\eta}(\tau,s_1)\cdots\\
        &\quad\cdots\Phi_{\eta\eta}(\tau,s_2)\bar{\Gamma}(s_2) d\tau ds_1ds_2\succeq\frac{\alpha_1}{\lambda_{\max}(R)}\cdots\\
        &\quad e^{-2\delta c_\eta}\int_{t}^{t+\delta}\int_t^{\tau}\int_{t}^{\tau}\bar{\Gamma}^\top(s_1)\bar{\Gamma}(s_2)d\tau ds_1ds_2\\
        &\succeq\frac{\alpha_1}{\lambda_{\max}(R)}e^{-2\delta c_\eta}\frac{(\gamma_1/\alpha_2)\delta^3}{3} I,\\
        \left\Vert\mathcal{O}_{\eta b}\right\Vert&=\left\Vert\int_t^{t+\delta}\Phi_{\eta\eta}^\top(\tau,t)H_\eta^\top\psi^{\top}_\tau R_\tau^{-1}\psi_\tau H_\eta\Phi_{\eta b}(\tau,t)d\tau\right\Vert\\
        &\le\frac{\alpha_2\beta e^{2\delta c_\eta}}{\lambda_{\min}(R)\sqrt{\alpha_1}} \int_t^{t+\delta}(\tau-t)d\tau=\frac{\alpha_2\delta^2\beta e^{2\delta c_\eta}}{2\lambda_{\min}(R)\sqrt{\alpha_1}}.
    \end{align*}
    A sufficient condition for the positive-definiteness of $\mathcal{O}$ is $\lambda_{\min}(\mathcal{O}_{\eta\eta})\lambda_{\min}(\mathcal{O}_{bb})>\Vert\mathcal{O}_{\eta b}\Vert^2$, that is,
    \begin{equation*}
        \frac{\alpha_1\delta}{\lambda_{\max}(R)} e^{-2\delta c_\eta}\frac{\alpha_1\gamma_1\delta^3}{3\alpha_2\lambda_{\max}(R)} e^{-2\delta c_\eta}>\frac{\alpha^2_2\delta^4\beta^2 e^{4c_\eta\delta}}{4\alpha_1\lambda^2_{\min}(R)},
    \end{equation*}
    which is exactly the excitation condition (4). We have proved uniform obseravbility for the transformed system. Since the transformation is uniformly bounded, so is the original system. 
\endproof

\subsection{Proof of Theorem~\ref{theorem::stability_bias}}\label{app::proof_stability_bias}
    Consider a Type-1 system with bias. Since the true $\chi$ evolves in a compact set $\mathcal{G}_1$, $z:=\pi(\chi)$, where $z_j^{(i)}:=\chi^{-1}\tilde{A}^jd^{(i)}$, lies in a compact set $\mathcal{Z}$ of the embedded space $\mathbb{R}^{Md_y^2}$. By Assumption~\ref{assumption::structure_rank_condition}, we have $d(\hat{\chi},\chi)\le c_1\Vert\hat{Z}-Z\Vert$ for some $c_1>0$. It suffices to show that there exists a compact subset $\hat{\mathcal{Z}}\times\hat{\mathcal{B}}\subset\mathbb{R}^{Md_y^2}\times\mathbb{R}^{\dim\mathfrak{g}}$ such that $[\hat{z}^\top,\hat{b}^\top]^\top$, initialized in $\text{int}(\mathcal{Z}\times\mathcal{B})$, with $P(t_0)=P_0$, remains in $\hat{\mathcal{Z}}\times\hat{\mathcal{B}}$ and that $[\hat{z}^\top,\hat{b}^\top]^\top$ converges exponentially to $[z^\top,b^\top]^\top$ after some finite time. To apply \cite[Th. 1]{semi_global_ekf} and verify its conditions, we verify that (1) the Jacobians are uniformly observable, and thus persistently determinable since the transition matrix is uniformly bounded; (2) $[z(t)^\top,b^\top]^\top$ is bounded; (3) the linearization error of the dynamics is given by $\varphi([z^\top,b^\top]^\top,[\hat{z}^\top,\hat{b}^\top]^\top):=F_u(z-\hat{z})-F_bz+F_{\hat{b}}\hat{z}-\breve{F}[(\hat{z}-z)^\top,(\hat{b}-b)^\top]^\top$, where the matrices are from \eqref{eq::embedded_system_bias} and $\breve{F}$ denotes the Jacobian, thus $\Vert\varphi\Vert$ is bounded by a quadratic function of $\Vert\hat{z}-z\Vert$ and $\Vert\hat{b}-b\Vert$. By \cite[Th. 1]{semi_global_ekf}, we obtain a semi-global stability result in joint bias estimation.
\endproof

\subsection{Proof of Lemma~\ref{lm::umeyama}}\label{app::proof_lm_umeyama}
    Let the cost function be $J(R,W)$. We first decouple the vector part from $J$ by completing the square, as follows:
    \begin{align*}
        J&=\tr\left[(\hatbar{Z}-R^{-1}(\bar D-W\underline{D}))(\hatbar{Z}-R^{-1}(\bar D-W\underline{D}))^\top\right]\\
        &=\tr\bigl[(\hatbar{Z}-R^{-1}\bar{D})(\hatbar{Z}-R^{-1}\bar{D})^\top+2R^{-1}W\underline{D}(\hatbar{Z}\\
        &\quad -R^{-1}\bar{D})^\top+R^{-1}W\underline{D}\underline{D}^\top W^\top R\bigl]\\
        &=\tr\biggl[\left(R^{-1}W\underline{D}\underline{D}^\top+(\hatbar{Z}-R^{-1}\bar{D})\underline{D}^\top\right)(\underline{D}\underline{D}^\top)^{-1}\\
        &\quad \left(R^{-1}W\underline{D}\underline{D}^\top+(\hatbar{Z}-R^{-1}\bar{D})\underline{D}^\top\right)^\top\biggl]+\tr\biggl[\left(\hatbar{Z}-R^{-1}\bar{D}\right)\\
        &\quad \left(I_{MN\times MN}-\underline{D}^\top(\underline{D}\underline{D}^\top)^{-1}\underline{D}\right)\left(\hatbar{Z}-R^{-1}\bar{D}\right)^\top\biggl].
    \end{align*}
    Since the first term is nonnegative, the global minimum $(R^*,W^*)$ of $J$ is achieved when
    \begin{align}
        \label{eq::decouple_condition}
        W^*=(\bar{D}-R^*\hatbar{Z})\underline{D}^\top(\underline{D}\underline{D}^\top)^{-1}.
    \end{align}
    This implies that we can minimize the second term with repect to $R$ only, and substitute the $R$ that achieves the global minimum of the second term  into \eqref{eq::decouple_condition}. Let $\bar{L}\bar{L}^\top=I_{MN\times MN}-\underline{D}^\top(\underline{D}\underline{D}^\top)^{-1}\underline{D}$ be the Cholesky decomposition. It suffices to consider the optimization problem
    \begin{align}
        \label{eq::rotation_only_opt}
        \mathop{\min}_{R\in\SO(d)}\left\Vert\hatbar{Z}\bar{L}-R^{-1}\bar{D}\bar{L}\right\Vert^2.
    \end{align}
    Expanding $\tilde J$ using the properties of the trace yields
    \begin{align}
        \tilde{J}=\Vert\hatbar{Z}\bar{L}\Vert^2+\Vert\hatbar{D}\bar{L}\Vert^2-2\tr\left(\hatbar{Z}\bar{L}\bar{L}^\top\bar{D}^\top R\right).
    \end{align}
    By the conditions of Lemma~\ref{lm::umeyama}, let $\bar{U}\Lambda\bar{V}^\top$ be a singular value decomposition of $\hatbar{Z}\bar{L}\bar{L}^\top\bar{D}^\top$ with singular values $\Lambda:=\text{diag}(\sigma_1,...,\sigma_d)$ in decreasing order. Using the results of \cite{umeyama_algorithm}, the global optimum $R^*$ of \eqref{eq::rotation_only_opt}, and thus \eqref{eq::umeyama_optimization} is given by
    \begin{equation}
        \label{eq::optimum_rotation}
        R^*=\bar{V}\bar{S}\bar{U}^\top.
    \end{equation}
    Substituting \eqref{eq::optimum_rotation} into \eqref{eq::decouple_condition}, we obtain $W^*=(\bar{D}-\bar{V}\bar{S}\bar{U}^\top\hatbar{Z})\underline{D}^\top(\underline{D}\underline{D}^\top)^{-1}$, which completes the proof.
\endproof

\subsection{Proof of Proposition~\ref{prop::vector_bias}}\label{app::proof_vector_bias}
    If $b_\omega=0$, define $b_\rho^{(i,j)}:=b_\rho\underline{d}_j^{(i)}\in\mathbb{R}^{d}$, the embedded system becomes
    \begin{equation}
        \label{eq::vector_bias_only_sys}
        \left\{\begin{aligned}
        \dot{\bar{z}}^{(i)}_j&=-\omega_t^\times\bar{z}^{(i)}_j-\rho_t\underline{d}_j^{(i)}-b_\rho^{(i,j)}-\bar{z}_{j+1}^{(i)}\\
        \dot{\bar{z}}_{N-1}^{(i)}&=-\omega_t^\times\bar{z}^{(i)}_{N-1}-\rho_t\underline{d}_{N-1}^{(i)}-b_\rho^{(i,N-1)}-\sum_{l=0}^{N-1}\tilde a_l\bar{z}_l^{(i)}\\
        \dot{b}^{(i,j)}_\rho&=0,\ \bar{y}^{(i)}=\bar{z}_0^{(i)},\ i\in[1,M],\ j\in[0,N-2]
        \end{aligned}\right.,
    \end{equation}
    which remains an LTV system. As in the proof of Lemma~\ref{lm::FH_observability}, the off-diagonal part of \eqref{eq::vector_bias_only_sys} is Kalman observable, and the diagonal part is uniformly observable. Hence, \eqref{eq::vector_bias_only_sys} is uniformly observable. The Kalman observer for this system is thus globally exponentially stable. The full rank of $D$ enables the reconstruction of the $\TFG(d,n,m)$ state. For the bias, using the estimates $\hat{b}_{\rho}^{(i,j)}$, we have $b_\rho[...\underline{d}_j^{(i)}...]=[...\hat{b}_{\rho}^{(i,j)}...]$. Since the column rank of $D$ is full, $\text{rank}([...\underline{d}_j^{(i)}...])=n+m$. We can obtain an estimate of the bias through solving this linear equation. As in the proof of Theorem~\ref{theorem::stability_non_bias}, the bias estimation error converges to zero exponentially.
\endproof

\subsection{Proof of Proposition~\ref{proposition::range_immersion}}\label{app::proof_range_immersion}
    It suffices to show that differentiating $s_{j,k}^{(i)}$ does not produce undesirable terms that are not in $\bar{z}$ or $s$. First, for $0\le j\le k\le N-2$, we have $
    \dot{s}^{(i)}_{j,k}=\frac{1}{2}\dot{\bar{z}}_j^{(i)\top}\bar{z}^{(i)}_k+\frac{1}{2}\bar{z}_j^{(i)\top}\dot{\bar{z}}^{(i)}_k=\frac{1}{2}(-\omega_t^\times\bar{z}^{(i)}_j-\rho_t\underline{d}^{(i)}_j-\bar{z}^{(i)}_{j+1})^\top\bar{z}^{(i)}_k+\frac{1}{2}\bar{z}_j^{(i)\top}(-\omega_t^\times\bar{z}^{(i)}_k-\rho_t\underline{d}^{(i)}_k-\bar{z}^{(i)}_{k+1})=-\frac{1}{2}(\rho_t\underline{d}^{(i)}_j)^{\top}\bar{z}^{(i)}_k-s^{(i)}_{j+1,k}-\frac{1}{2}(\rho_t\underline{d}_k^{(i)})^\top\bar{z}^{(i)}_j-s^{(i)}_{j,k+1}$. Due to the skew-symmetric structure of $\omega_t^\times$, the terms $\frac{1}{2}(-\omega_t^\times\bar{z}_j^{(i)})^\top\bar{z}_k^{(i)}$ and $\frac{1}{2}\bar{z}_j^{(i)\top}(-\omega_t^\times\bar{z}_k^{(i)})$ cancel. Otherwise, these two terms and their subsequent derivatives would produce many new terms, preventing finite termination. The same reasoning applies to the cases $0\le j\le k=N-1$ and $j=k=N-1$. The coefficients $\tilde{a}_\iota$ arise from the dynamics $f_1$.
    
    Thus, the embedding of $f_2$ in \eqref{eq::ltv_range_immersion} for two-frame systems \eqref{eq::tfg_case1_systems} with range measurements is given by $\dot{s}^{(i)}_{j,k}=-\frac{1}{2}(\rho_t\underline{d}^{(i)}_j)^{\top}\bar{z}^{(i)}_k-s^{(i)}_{j+1,k}-\frac{1}{2}(\rho_t\underline{d}_k^{(i)})^\top\bar{z}^{(i)}_j-s^{(i)}_{j,k+1}$, for $0\le j\le k\le N-2$; $\dot{s}^{(i)}_{j,N-1}=-\frac{1}{2}(\rho_t\underline{d}_j^{(i)})^\top\bar{z}^{(i)}_{N-1}-s^{(i)}_{j+1,N-1}-\frac{1}{2}(\rho_t\underline{d}_{N-1}^{(i)})^\top\bar{z}^{(i)}_j-\sum_{\iota=0}^{N-1}\tilde{a}_\iota s_{j,\iota}$, for $0\le j\le N-2$; and $\dot{s}^{(i)}_{N-1,N-1}=-(\rho_t\underline{d}_{N-1}^{(i)})^\top\bar{z}^{(i)}_{N-1}-\sum\limits_{\iota=0}^{N-1}2\tilde{a}_\iota s_{\iota,N-1}$. 
\endproof


\section*{References}
\bibliographystyle{IEEEtran}
\bibliography{IEEEabrv, refs.bib}

\begin{thebibliography}{10}
\providecommand{\url}[1]{#1}
\csname url@rmstyle\endcsname
\providecommand{\newblock}{\relax}
\providecommand{\bibinfo}[2]{#2}
\providecommand\BIBentrySTDinterwordspacing{\spaceskip=0pt\relax}
\providecommand\BIBentryALTinterwordstretchfactor{4}
\providecommand\BIBentryALTinterwordspacing{\spaceskip=\fontdimen2\font plus
\BIBentryALTinterwordstretchfactor\fontdimen3\font minus
  \fontdimen4\font\relax}
\providecommand\BIBforeignlanguage[2]{{%
\expandafter\ifx\csname l@#1\endcsname\relax
\typeout{** WARNING: IEEEtran.bst: No hyphenation pattern has been}%
\typeout{** loaded for the language `#1'. Using the pattern for}%
\typeout{** the default language instead.}%
\else
\language=\csname l@#1\endcsname
\fi
#2}}

\bibitem{LoS}
A.~Barrau and S.~Bonnabel, ``Linear observed systems on groups,'' \emph{Systems
  \& Control Letters}, vol. 129, pp. 36--42, July 2019.

\bibitem{PreIntMath}
A.~Barrau and S.~Bonnabel, ``A mathematical framework for imu error propagation
  with applications to preintegration,'' in \emph{2020 IEEE International
  Conference on Robotics and Automation (ICRA)}, 2020, pp. 5732--5738.

\bibitem{LoS_Mfd}
C.~Liu and Y.~Shen, ``On the existence of linear observed systems on manifolds
  with connection,'' \emph{IEEE Control Systems Letters}, vol.~8, pp.
  2607--2612, 2024.

\bibitem{TFG}
A.~Barrau and S.~Bonnabel, ``The geometry of navigation problems,'' \emph{IEEE
  Transactions on Automatic Control}, vol.~68, no.~2, pp. 689--704, 2023.

\bibitem{InEKF}
A.~Barrau and S.~Bonnabel, ``The invariant extended kalman filter as a stable
  observer,'' \emph{IEEE Transactions on Automatic Control}, vol.~62, no.~4,
  pp. 1797--1812, 2017.

\bibitem{AnnuRevInEKF}
A.~Barrau and S.~Bonnabel, ``Invariant kalman filtering,'' \emph{Annual Review
  of Control, Robotics, and Autonomous Systems}, vol.~1, pp. 237--257, 2018.

\bibitem{filtering_lie_groups}
A.~Barrau and S.~Bonnabel, ``Intrinsic filtering on lie groups with
  applications to attitude estimation,'' \emph{IEEE Transactions on Automatic
  Control}, vol.~60, no.~2, pp. 436--449, 2015.

\bibitem{rotating_earth}
M.~Brossard, A.~Barrau, P.~Chauchat, and S.~Bonnabel, ``Associating uncertainty
  to extended poses for on lie group imu preintegration with rotating earth,''
  \emph{IEEE Transactions on Robotics}, vol.~38, no.~2, pp. 998--1015, 2022.

\bibitem{EqVIO}
P.~van Goor and R.~Mahony, ``Eqvio: An equivariant filter for visual-inertial
  odometry,'' \emph{IEEE Transactions on Robotics}, vol.~39, no.~5, pp.
  3567--3585, 2023.

\bibitem{MSCEqF}
A.~Fornasier, P.~van Goor, E.~Allak, R.~Mahony, and S.~Weiss, ``Msceqf: A multi
  state constraint equivariant filter for vision-aided inertial navigation,''
  \emph{IEEE Robotics and Automation Letters}, vol.~9, no.~1, pp. 731--738,
  2024.

\bibitem{InGVIO}
C.~Liu, C.~Jiang, and H.~Wang, ``Ingvio: A consistent invariant filter for fast
  and high-accuracy gnss-visual-inertial odometry,'' \emph{IEEE Robotics and
  Automation Letters}, vol.~8, no.~3, pp. 1850--1857, 2023.

\bibitem{ExploitSymm}
M.~Brossard, A.~Barrau, and S.~Bonnabel, ``Exploiting symmetries to design ekfs
  with consistency properties for navigation and slam,'' \emph{IEEE Sensors
  Journal}, vol.~19, no.~4, pp. 1572--1579, 2019.

\bibitem{InEKF_wheel_gnss_arm}
P.~Chauchat, A.~Barrau, and S.~Bonnabel, ``Invariant filtering for wheeled
  vehicle localization with unknown wheel radius and unknown gnss lever arm,''
  in \emph{2024 IEEE 63rd Conference on Decision and Control (CDC)}, 2024, pp.
  2005--2011.

\bibitem{IEKF-wind}
Z.~Ahmed and C.~A. Woolsey, ``An invariant extended kalman filter for wind
  estimation using a small, fixed-wing uncrewed aerial vehicle,'' \emph{IEEE
  Transactions on Control Systems Technology}, vol.~33, no.~5, pp. 1799--1809,
  2025.

\bibitem{obs_design_nls}
P.~Bernard, \emph{Observer Design for Nonlinear Systems}.\hskip 1em plus 0.5em
  minus 0.4em\relax Cham, Switzerland: Springer, 2018.

\bibitem{EqF}
P.~van Goor, T.~Hamel, and R.~Mahony, ``Equivariant filter (eqf),'' \emph{IEEE
  Transactions on Automatic Control}, vol.~68, no.~6, pp. 3501--3512, 2023.

\bibitem{Annurev_EqF}
R.~Mahony, P.~van Goor, and T.~Hamel, ``Observer design for nonlinear systems
  with equivariance,'' \emph{Annual Review of Control, Robotics, and Autonomous
  Systems}, vol.~5, pp. 221--252, 2022.

\bibitem{synchronous_observer_design}
P.~{van Goor}, P.~Wickramasinghe, M.~Hampsey, and R.~Mahony, ``Constructive
  synchronous observer design for inertial navigation with delayed gnss
  measurements,'' \emph{European Journal of Control}, vol.~80, p. 101047, 2024,
  2024 European Control Conference Special Issue.

\bibitem{AutoErrorGroupAffine}
P.~van Goor and R.~Mahony, ``Autonomous error and constructive observer design
  for group affine systems,'' in \emph{2021 60th IEEE Conference on Decision
  and Control (CDC)}, 2021, pp. 4730--4737.

\bibitem{topological_obstruction}
S.~P. Bhat and D.~S. Bernstein, ``A topological obstruction to continuous
  global stabilization of rotational motion and the unwinding phenomenon,''
  \emph{Systems \& Control Letters}, vol.~39, no.~1, pp. 63--70, 2000.

\bibitem{cf_att}
R.~Mahony, T.~Hamel, and J.-M. Pflimlin, ``Nonlinear complementary filters on
  the special orthogonal group,'' \emph{IEEE Transactions on Automatic
  Control}, vol.~53, no.~5, pp. 1203--1218, 2008.

\bibitem{improved_cf_att}
S.~Berkane and A.~Tayebi, ``On the design of attitude complementary filters on
  $\text{SO}(3)$,'' \emph{IEEE Transactions on Automatic Control}, vol.~63,
  no.~3, pp. 880--887, 2018.

\bibitem{hybrid_observer_landmark}
M.~Wang and A.~Tayebi, ``Hybrid nonlinear observers for inertial navigation
  using landmark measurements,'' \emph{IEEE Transactions on Automatic Control},
  vol.~65, no.~12, pp. 5173--5188, 2020.

\bibitem{hybrid_observer_vision_aided}
M.~Wang, S.~Berkane, and A.~Tayebi, ``Nonlinear observers design for
  vision-aided inertial navigation systems,'' \emph{IEEE Transactions on
  Automatic Control}, vol.~67, no.~4, pp. 1853--1868, 2022.

\bibitem{ins_observer_tutorial}
M.~Wang and A.~Tayebi, ``Observers design for inertial navigation systems: A
  brief tutorial,'' in \emph{2020 59th IEEE Conference on Decision and Control
  (CDC)}, 2020, pp. 1320--1327.

\bibitem{att_scalar}
H.~Alnahhal, S.~Benahmed, S.~Berkane, and T.~Hamel, ``Attitude estimation using
  scalar measurements,'' \emph{IEEE Control Systems Letters}, vol.~9, pp.
  1862--1867, 2025.

\bibitem{body_frame_state_estimation}
S.~Benahmed and S.~Berkane, ``State estimation using single body-frame bearing
  measurements,'' in \emph{2024 European Control Conference (ECC)}, 2024, pp.
  116--121.

\bibitem{geometric_approach_imu_body}
S.~Benahmed, S.~Berkane, and T.~Hamel, ``A geometric approach for pose and
  velocity estimation using imu and inertial/body-frame measurements,'' in
  \emph{2025 IEEE 64th Conference on Decision and Control (CDC)}, 2025, pp.
  7425--7430.

\bibitem{ins_ltv}
\BIBentryALTinterwordspacing
S.~Benahmed, T.~Hamel, and S.~Berkane, ``A general nonlinear observer design
  for inertial navigation systems with almost global stability guarantees,''
  2026. [Online]. Available: \url{https://arxiv.org/abs/2410.03846}
\BIBentrySTDinterwordspacing

\bibitem{riccati_observer_pnp}
T.~Hamel and C.~Samson, ``Riccati observers for the nonstationary pnp
  problem,'' \emph{IEEE Transactions on Automatic Control}, vol.~63, no.~3, pp.
  726--741, 2018.

\bibitem{GTM218}
J.~M. Lee, \emph{Introduction to Smooth Manifolds}.\hskip 1em plus 0.5em minus
  0.4em\relax New York: Springer, 2013.

\bibitem{GTM222}
B.~C. Hall, \emph{Lie Groups, Lie Algebras, and Representations An Elementary
  Introduction}.\hskip 1em plus 0.5em minus 0.4em\relax Cham, Switzerland:
  Springer, 2015.

\bibitem{linear_systems_theory}
C.-T. Chen, \emph{Linear System Theory and Design}, 3rd~ed.\hskip 1em plus
  0.5em minus 0.4em\relax New York, NY, USA: Oxford University Press, 1999.

\bibitem{boundedness_riccati_ode}
M.~Pengov, E.~Richard, and J.-C. Vivalda, ``On the boundedness of the solutions
  of the continuous riccati equation,'' \emph{Journal of Inequalities and
  Applications}, vol.~6, no.~6, pp. 641--649, 2001.

\bibitem{nonlinear_immersion2}
G.~Besancon and A.~Ticlea, ``An immersion-based observer design for
  rank-observable nonlinear systems,'' \emph{IEEE Transactions on Automatic
  Control}, vol.~52, no.~1, pp. 83--88, 2007.

\bibitem{umeyama_algorithm}
S.~Umeyama, ``Least-squares estimation of transformation parameters between two
  point patterns,'' \emph{IEEE Transactions on Pattern Analysis and Machine
  Intelligence}, vol.~13, no.~4, pp. 376--380, 1991.

\bibitem{EqFCDC}
P.~van Goor, T.~Hamel, and R.~Mahony, ``Equivariant filter (eqf): A general
  filter design for systems on homogeneous spaces,'' in \emph{2020 59th IEEE
  Conference on Decision and Control (CDC)}, 2020, pp. 5401--5408.

\bibitem{nonlinear_immersion1}
J.~Levine and R.~Marino, ``Nonlinear system immersion, observers and
  finite-dimensional filters,'' \emph{Systems \& Control Letters}, vol.~7,
  no.~2, pp. 133--142, 1986.

\bibitem{EqFBias}
A.~Fornasier, Y.~Ng, R.~Mahony, and S.~Weiss, ``Equivariant filter design for
  inertial navigation systems with input measurement biases,'' in \emph{2022
  International Conference on Robotics and Automation (ICRA)}, 2022, pp.
  4333--4339.

\bibitem{af_thesis}
A.~Fornasier, Y.~Ge, P.~{van Goor}, R.~Mahony, and S.~Weiss, ``Equivariant
  symmetries for inertial navigation systems,'' \emph{Automatica}, vol. 181, p.
  112495, 2025.

\bibitem{Morse_theory}
J.~W. Milnor, \emph{Morse Theory}.\hskip 1em plus 0.5em minus 0.4em\relax New
  Jersey, USA: Princeton University Press, 1969.

\bibitem{semi_global_ekf}
P.~Bernard, N.~Mimmo, and L.~Marconi, ``On the semi-global stability of an
  ek-like filter,'' \emph{IEEE Control Systems Letters}, vol.~5, no.~5, pp.
  1771--1776, 2021.

\bibitem{ekf_modified_riccati_ode}
K.~Reif, F.~Sonnemann, and R.~Unbehauen, ``An ekf-based nonlinear observer with
  a prescribed degree of stability,'' \emph{Automatica}, vol.~34, no.~9, pp.
  1119--1123, 1998.

\bibitem{pose_velocity_landmark_position_wang}
M.~Wang and A.~Tayebi, ``Nonlinear attitude estimation using intermittent
  linear velocity and vector measurements,'' in \emph{Proc. of the 60th IEEE
  Conference on Decision and Control}, 2021, pp. 4707--4712.

\bibitem{bearing_ins}
T.~Hamel, M.-D. Hua, and C.~Samson, ``Deterministic observer design for
  vision-aided inertial navigation,'' in \emph{2020 59th IEEE Conference on
  Decision and Control (CDC)}, 2020, pp. 1306--1313.

\bibitem{bearing_ins_acc_bias}
S.~de~Marco, M.-D. Hua, T.~Hamel, and C.~Samson, ``Position, velocity, attitude
  and accelerometer-bias estimation from imu and bearing measurements,'' in
  \emph{2020 European Control Conference (ECC)}, 2020, pp. 1003--1008.

\bibitem{relative_pose_bearing}
M.-D. Hua, S.~De~Marco, T.~Hamel, and R.~W. Beard, ``Relative pose estimation
  from bearing measurements of three unknown source points,'' in \emph{2020
  59th IEEE Conference on Decision and Control (CDC)}, 2020, pp. 4176--4181.

\bibitem{LTI_quadratic_outputs}
S.~Berkane, D.~Theodosis, T.~Hamel, and D.~V. Dimarogonas, ``State estimation
  for linear systems with quadratic outputs,'' \emph{IEEE Control Systems
  Letters}, vol.~7, pp. 3872--3877, 2023.

\bibitem{LTI_range_navigation}
L.~Rodrigues, ``On the relation between observability of a class of linear
  systems with norm outputs and navigation with range measurements,''
  \emph{IEEE Control Systems Letters}, vol.~9, pp. 1243--1248, 2025.

\bibitem{position_range_riccati}
T.~Hamel and C.~Samson, ``Position estimation from direction or range
  measurements,'' \emph{Automatica}, vol.~82, pp. 137--144, 2017.

\end{thebibliography}

\begin{IEEEbiography}
[{\includegraphics[width=1in, height=1.25in, clip, keepaspectratio]{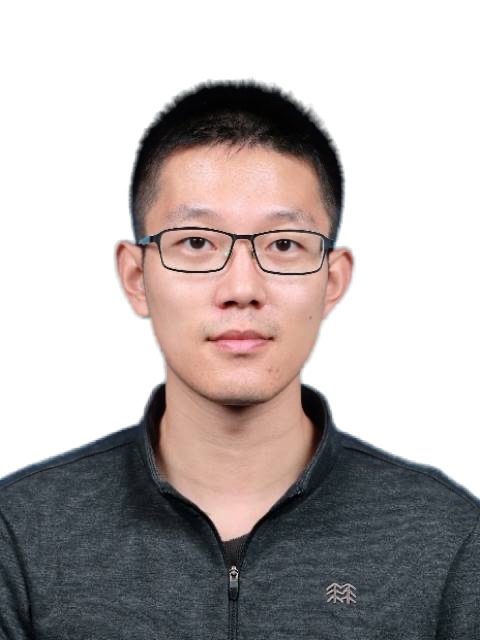}}]
    {Changwu Liu} received the B.E. degree in mechanics and the Ph.D. degree in aerospace engineering from Tsinghua University, Beijing, China, in 2018 and 2024, respectively.
	
    He is currently a Post-Doctoral Researcher in the Department of Electronic Engineering, Tsinghua University. His research interests include state estimation, nonlinear observers, geometric control methods and their various applications in robotics.
\end{IEEEbiography}

\vspace{-0.7cm}

\begin{IEEEbiography}
[{\includegraphics[width=1in, height=1.25in, clip, keepaspectratio]{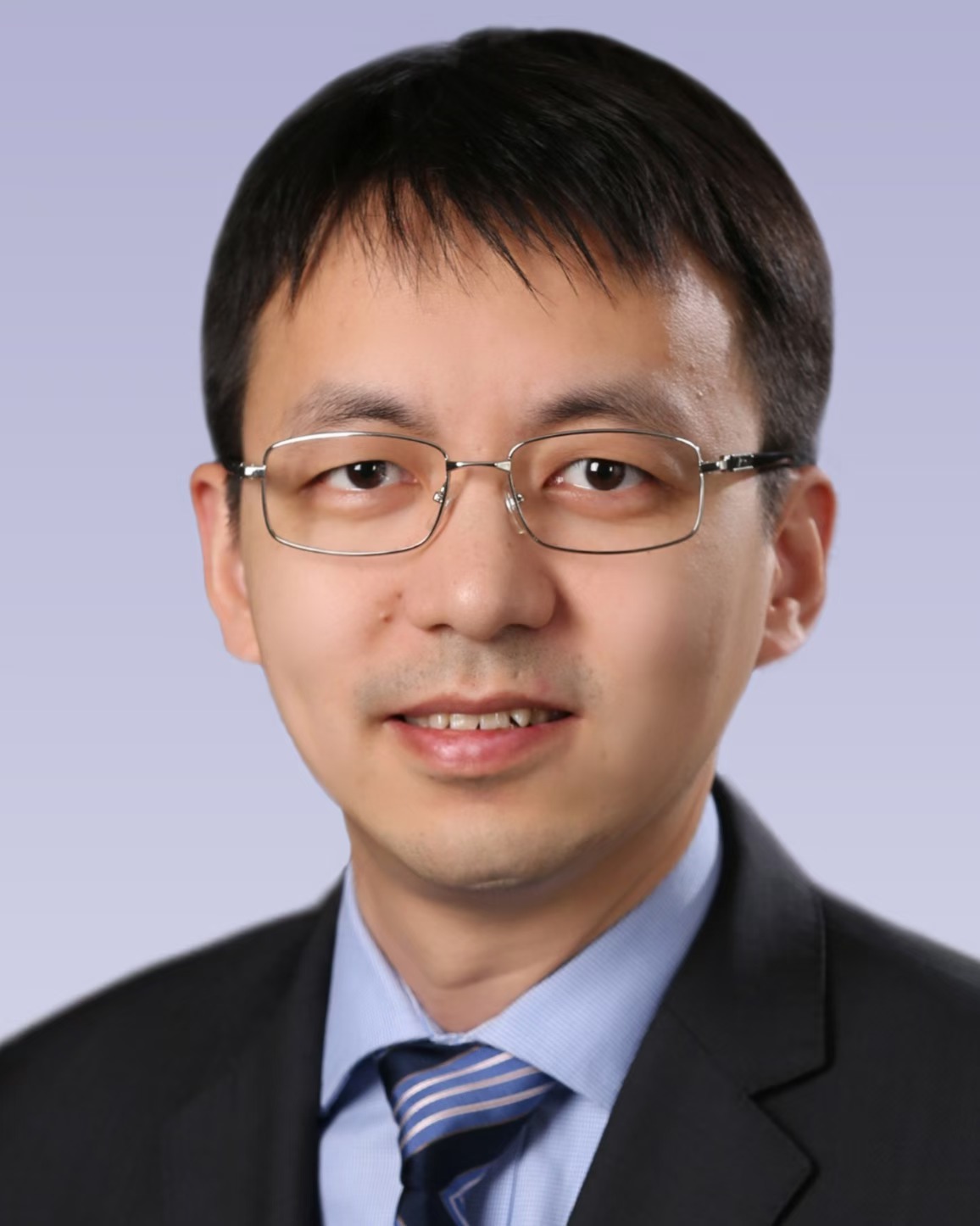}}]
    {Yuan Shen} (Senior Member, IEEE) received the B.E. degree in electronic engineering from Tsinghua University in 2005 and the S.M. and Ph.D. degrees in electrical engineering and computer science from Massachusetts Institute of Technology (MIT) in 2008 and 2014, respectively. 

    He is currently a Full Professor with the Department of Electronic Engineering, Tsinghua University. His research interests include network localization and navigation, integrated sensing and control, and multi-agent systems. His papers have received the IEEE ComSoc Fred W. Ellersick Prize and several best paper awards from IEEE conferences. He has served as the TPC Symposium Co-Chair for IEEE ICC and IEEE Globecom for several times. He was the Elected Chair of the IEEE ComSoc Radio Communications Committee from 2019 to 2020. He is currently an Editor of the {\textsc{IEEE Transactions on Signal Processing}}, \textsc{IEEE Transactions on Communications}, \textsc{IEEE Transactions on Network Science and Engineering}, and \textit{China Communications.}
\end{IEEEbiography}

\end{document}